\documentclass[journal,twoside,web, 9pt]{ieeecolor}
\usepackage{multirow}
\usepackage{array}
\usepackage{generic}
\usepackage{cite}
\usepackage{amsmath,amssymb,amsfonts}
\usepackage{algorithm}
\usepackage{algorithmic}
\usepackage{graphicx}
\usepackage{algorithm,algorithmic}
\usepackage{hyperref}

\usepackage{balance}

\hypersetup{hidelinks}
\usepackage{textcomp}
\def\BibTeX{{\rm B\kern-.05em{\sc i\kern-.025em b}\kern-.08em
    T\kern-.1667em\lower.7ex\hbox{E}\kern-.125emX}}
\begin{document}
\title{Distributionally robust two-stage model predictive control$:$ adaptive constraint tightening with stability guarantee}

\author{
Weijiang Zheng, Jiayi Huang, and Bing Zhu, \IEEEmembership{Senior Member, IEEE},
\thanks{The authors are with The Seventh Research Division, Beihang University, Beijing 100191, P.R.China. (E-mails: zhengweijiang@buaa.edu.cn; 2503021@buaa.edu.cn; zhubing@buaa.edu.cn) }
}

\maketitle

\begin{abstract}
This paper proposes a two-stage distributionally robust model predictive control (TSDR-MPC) scheme for stochastic disturbances with unknown time-varying means and covariances. By defining a Wasserstein ambiguity set on the disturbance‑to‑constraint space, constraint violation penalties are formulated as a second-stage problem, enabling adaptive tightening. A finitely convergent cutting-plane algorithm is developed for real-time implementation. The framework naturally degrades to deterministic MPC as uncertainty vanishes, without pre‑specified tightening parameters. Theoretical guarantees include feasibility, finite-time termination, and an asymptotic average cost bound. Numerical simulations validate its adaptability and robustness.
\end{abstract}


\section{Introduction}
\label{sec:introduction}

Model predictive control (MPC) is an optimization-based control strategy that has garnered significant attention from both academia and industry over the past few decades, primarily due to its ability to explicitly handle system constraints \cite{mark2020stochastic}. In practical applications, system states are often affected by disturbances, which may originate from measurement noise, process noise, or model-plant mismatch. To address these uncertainties, researchers have developed two main frameworks, robust MPC and stochastic MPC.

Robust MPC assumes that disturbances are bounded and guarantees constraint satisfaction under worst-case scenarios. However, this approach tends to be overly conservative, as the probability of worst-case occurrences is extremely low \cite{mayne2005robust}. Stochastic MPC, on the other hand, leverages probabilistic information about the disturbance distribution and allows for constraint violations with a certain probability through chance constraints, thereby seeking a balance between conservatism and performance \cite{mesbah2016stochastic,mayne2016robust}. Nevertheless, a key assumption of Stochastic MPC is that the exact probability distribution of the disturbance is known, which is often difficult to satisfy in practice. In real-world systems, we can only estimate the distribution from limited historical data, and estimation errors may lead to constraint violations or performance degradation.

To address this challenge, Distributionally Robust Optimization has been introduced into the MPC framework, giving rise to distributionally robust model predictive control (DRMPC). The core idea of DRMPC is that instead of assuming the disturbance follows an exact distribution, we define an ambiguity set that contains all possible distributions and optimize performance under the worst-case distribution within that set \cite{mohajerin2018data}. This approach not only utilizes statistical information from data but also provides robustness against distributional estimation errors. 
Furthermore, two-stage distributionally robust optimization methods developed in the stochastic programming literature are characterized by the distinction between deterministic here-and-now first-stage decisions and worst-case expectation wait-and-see second-stage decisions \cite{doi:10.1137/20M1370227,hanasusanto2018conic}. Although the theoretical contributions of two-stage DRO provide important foundations for distributionally robust optimization, their application within the MPC framework remains relatively limited, and further exploration in this direction merits attention.

The construction of ambiguity sets is a central theme in DRMPC literature, with two predominant approaches being moment-based sets and Wasserstein distance-based sets. Moment-based ambiguity sets, which contain all distributions sharing specified moments such as mean and covariance, have been widely employed to handle both bounded and unbounded disturbances \cite{li2021distributionally,zhang2025distributionally,li2023distributionally,mark2022recursively,li2024distributionally}. For systems with unbounded disturbances, moment-based method can reformulate chance constraints into tractable convex forms such as second-order cone and obtain recursive feasibility and stability guarantees. In contrast, Wasserstein distance-based ambiguity sets have gained significant popularity due to their rich geometric structure and strong out-of-sample performance guarantees \cite{mark2020stochastic,zhong2022distributionally,fochesato2022data,micheli2022data,coulson2021distributionally,pilipovsky2024distributionally,hakobyan2024wasserstein,hakobyan2021wasserstein,aolaritei2023wasserstein}. These sets are particularly valuable when only limited data are available, as they provide probabilistic confidence that the true distribution lies within the ambiguity set. Numerous works have used Wasserstein ambiguity sets to propagate uncertainty through system dynamics and enforce distributionally robust chance constraints or Conditional Value-at-Risk constraints.

A distinct architectural approach within DRMPC is tube-based model predictive control, which extends the classic robust tube MPC concept to the stochastic distributionally robust setting. In tube-based DRMPC, a sequence of tubes, often constructed as distributionally robust probabilistic reachable sets, contains the error dynamics for all distributions within the ambiguity set \cite{mark2020stochastic,fochesato2022data,pilipovsky2024distributionally,aolaritei2023wasserstein}. This approach has been developed using Wasserstein ambiguity sets, providing a systematic way to tighten constraints around a nominal trajectory while ensuring constraint satisfaction for all plausible disturbance distributions. Beyond these main categories, the literature also explores other types of ambiguity sets such as total variation distance\cite{dixit2022distributionally,zolanvari2025iterative} and divergence-based sets \cite{schuurmans2023general}, which offer alternative ways to quantify distributional uncertainty. Risk measures, particularly Conditional Value-at-Risk, are frequently integrated with various ambiguity set constructions to encode safety constraints in a distributionally robust manner \cite{micheli2022data,coulson2021distributionally,pilipovsky2024distributionally,aolaritei2023wasserstein,zolanvari2025iterative}. Across these diverse approaches, theoretical properties such as recursive feasibility and closed-loop stability have been rigorously established.

Despite these advances, most existing DRMPC methods rely on the assumption that the disturbance has zero mean or that its moments are known a priori. In many practical scenarios, however, disturbances may exhibit unknown non-zero means that varies over time, and their covariance may also be time-varying.  
In this paper, we propose a two-stage distributionally robust MPC framework that explicitly addresses uncertain time-varying means and covariances of disturbances. 
The main contributions are as follows:

1) Two‑stage adaptive constraint tightening: Constraint violations are penalized via a second‑stage linear program whose dual variables automatically adjust the tightening effect based on the state and sampled disturbances, requiring no pre‑specified parameters.

2) Handling unknown time‑varying moments: Only bounds on the first two moments are needed; neither exact values nor zero‑mean assumption is required, making the framework applicable to non‑stationary systems.

3) Tractable reformulation via Wasserstein ambiguity set: The minimax problem is reformulated as a set of finite‑dimensional convex programs, solved efficiently by a cutting‑plane algorithm.

4 Rigorous theoretical guarantees: Recursive feasibility, finite termination of the algorithm, and an asymptotic upper bound on the average cost are proved, explicitly revealing the trade‑off between robustness and performance.

Simulations under non‑zero mean, large covariance, and time‑varying disturbances show the controller adaptively adjusts its conservatism while maintaining stability and constraint satisfaction.

The remainder of the paper is organized as follows. Section \ref{sec2} presents the problem formulation and the construction of the Wasserstein ambiguity set. Section \ref{sec3} derives the tractable reformulation of the TSDR-MPC problem and introduces the terminal constraint and the cutting-plane algorithm. Section \ref{sec 4} provides the stability analysis and proves the asymptotic performance bound. Section \ref{sec 5} presents simulation results, and Section \ref{sec6} concludes the paper.

\section{Problem Statement}\label{sec2}

\subsection{System Dynamics}
Consider the discrete-time stochastic linear dynamics system
\begin{align}
x_{k+1} = Ax_k + Bu_k + Dw_k,
\end{align}
where $x_k \in \mathcal{X}\subset\mathbb{R}^{n_x}$ is the state, $u_k \in \mathcal{U}\subset\mathbb{R}^{n_u}$ is the control input, and $w_{k}\in \mathcal{W}\subset\mathbb{R}^{n_w}$ is a stochastic disturbance with an unknown distribution. The disturbance $w_{k}$ is independent across time steps $k$, but not necessarily identically distributed. $\mathcal  X$ and $\mathcal  U$ are compact sets with origin in their interior. We consider unbounded disturbances ${w}_k$ with $\mathcal{W}=\mathbb{R}^{n_w}$, assuming that its expectation $\mu_{k}$ and covariance matrix $\Sigma_{k}$ are uniformly bounded with uncertain constant bounds, i.e., $\left\| \mu_{k} \right\| \leq \bar{\mu},\ \Sigma_{k} \preceq \bar{\Sigma}$. For analysis and control design, we denote the system dynamics mapping function $f:\mathbb{R}^{n_x}\times\mathbb{R}^{n_u}\times\mathbb{R}^{n_w}\to\mathbb{R}^{n_x}$
\begin{align*}
    f(x,u,w)=Ax+Bu+Dw,
\end{align*}
and the system is rewritten in a compact form over a prediction horizon $N$,
\begin{align}
{\bar{x}}_{k} = \bar{A} x_{k} + \bar{B} {\bar{u}}_{k} + \bar{D} {\bar{w}}_{k},
\end{align}
where
\begin{align*}
\bar{A} &= \begin{bmatrix} A \\ \vdots \\ A^N \end{bmatrix}, \quad \bar{B} = \begin{bmatrix} B & 0 & \cdots & 0 \\ AB & B & \ddots & \vdots \\ \vdots & \vdots & \ddots & 0 \\ A^{N-1}B & A^{N-2}B & \cdots & B \end{bmatrix}, \\
\bar{D} &= \begin{bmatrix} D & 0 & \cdots & 0 \\ AD & D & \ddots & \vdots \\ \vdots & \vdots & \ddots & 0 \\ A^{N-1}D & A^{N-2}D & \cdots & D \end{bmatrix}.
\end{align*}
The N-step predicted state sequence $\bar{x}_k$, the input sequence $\bar{u}_k$ and the the unknown disturbance $\bar{w}_k$ from time $k$ are denoted by
\begin{equation}
\bar{x}_k = \begin{bmatrix}
    x_{1|k}\\ \vdots\\ x_{N|k}
\end{bmatrix},\bar{u}_k = \begin{bmatrix}u_{0|k}\\ \vdots\\ u_{N-1|k}\end{bmatrix},\bar{w}_k =\begin{bmatrix} w_{0|k}\\ \vdots\\ w_{N-1|k}\end{bmatrix}.
\end{equation}

The following assumption holds throughout the whole paper.
\newtheorem{assumption}{\bf Assumption} 

\begin{assumption}\label{assumption_dynamic_bound}

(a) $(A,B)$ is stabilizable and $(A,Q\textstyle^{1/2})$ is detectable, where $Q$ is positive semidefinite symmetric.
(b) There exist positive constants $L_A$, $L_B$, $L_D$, $u_u$, such that $\|A\|\leq L_A$, $\|B\|\leq L_B$, $\|D\|\leq L_D$ and $\|u-u'\|\leq u_u,\forall u,u'\in \mathcal{U}$.
\end{assumption}

Assumption \ref{assumption_dynamic_bound} provides two types of conditions: part (a) guarantees the existence of a stabilizing LQR controller [Proposition 4.4.1, \citenum{1995Dynamic}], while part (b) gives uniform bounds on the system matrices and control inputs.
In the presence of unbounded disturbances, \cite{11007013} performs stability analysis of MPC by comparing with an auxiliary controller, which inspires our approach. In the following, we will take the LQR controller as a baseline for our analysis.

The state constraints take the following inequality form,
\begin{align*}
F_{0}x_{i|k} + G_{0} \leq 0,
\end{align*}
where $F_{0}\in\mathbb{R}^{n_c\times n_x}$, $G_{0}\in\mathbb{R}^{n_c}$ and $n_c$ is the number of constraints. It is also written in a more compact form,
\begin{align}\label{eq_constraints}
    F{\bar{x}}_{k} + G \leq 0,
    \end{align}
where $F = \text{diag}\left( F_{0},\cdots,F_{0} \right)$, $G = \left\lbrack G_{0}^{T},\cdots,G_{0}^{T} \right\rbrack^{T}$.

\subsection{Ambiguity Set}
 In distributionally robust optimization, the worst-case is taken over an ambiguity set to find a decision that minimizes the worst-case expected cost. Moment-based ambiguity sets can handle only stationary processes and require precise moment estimates. To avoid these limitations, we employ an ambiguity set defined via the 2-Wasserstein distance (optimal transport distance),
\begin{align} \label{eq_wasserstein_distance}
W_{C}( \mathbb{P},\mathbb{Q} ) = \inf_{\mathbb{Z}}\{ \textstyle\int_{\mathcal W^N} c(w, w') \mathbb{Z}( dw,dw' ) d w \},
\end{align}
where $c(w,w') =\frac{1}{2} ( w - w')^{T} C_s \left(w -w' \right)$; $C_s=(F\bar D)^TCF\bar D$ and $C$ is a tunable weighting matrix which is symmetric positive definite; $\mathbb{Z}$ is a joint distribution of $w'$ and $w$ with marginal distributions $\mathbb{Q}$ and $\mathbb{P}$ on the support $\mathcal W^N$, respectively. The empirical distribution $\mathbb{Q}=\frac{1}{n}\sum_{\sigma=1}^n\delta_{w^\sigma}$ is chosen to generate an approximation of $\mathbb{P}$, where $\delta_{w^\sigma}$ denotes the Dirac delta distribution at ${w^\sigma}$. 
Let $\mathcal{M}(\mathcal{W}^N)$ be the set of all probability distributions supported on $\mathcal{W}^N$, we define the ambiguity set as
\begin{align}\label{eq_Wasserstein_ambiguity_set}
    \mathcal{P}_k=\{\mathbb{P}_k\in\mathcal{M}(\mathcal{W}^N)|W_{C}(\mathbb{P}_k,\mathbb{Q}_k)\leq \varepsilon\},
\end{align}
 where $\varepsilon>0$ is the radius of the ambiguity set.

\begin{assumption}\label{assumption:ambiguity set}
The zero distribution $\mathbb{O}$ and the true distribution $\mathbb{P}_{k}$ both belong to the ambiguity set defined by a Dirac delta empirical distribution $\mathbb{Q}_{k}$, i.e.,
\begin{align}
    W_{C}( \mathbb{P}_{k},\mathbb{Q}_{k} ) \leq \varepsilon, \quad W_{C}( \mathbb{O},\mathbb{Q}_{k}) \leq \varepsilon
\end{align}
\end{assumption}
\newtheorem{remark}{\bf Remark}
\begin{remark}
 Concentration inequalities [\citenum{2020Convergence}, Corollary 5.2] guarantee that the true distribution is contained with high probability within the ambiguity set defined around the empirical distribution. This guarantee serves as a fundamental postulate in distributionally robust optimization. For the trivial zero distribution, the condition can be met by using rough moment estimates and selecting an appropriate value for $\varepsilon$, which simplifies the ensuing stability analysis.
\end{remark}

{
}
To estimate the upper bounds of the worst-case distribution's moments, we adopt the Gelbrich bound in [\citenum{10509008}, Lemma 1]. The
following lower-bound holds for the 2-Wasserstein distance $W_{2}( \mathbb{P},\mathbb{Q} )$ ($W_{2}^2( \mathbb{P},\mathbb{Q} )=W_{C}( \mathbb{P},\mathbb{Q} )$ in case $C_s=2I$ in \eqref{eq_wasserstein_distance}),
\begin{align}\nonumber
G( \mathbb{P},\mathbb{Q} ) = \sqrt{( \mu - \hat{\mu} )^{T}( \mu - \hat{\mu} ) + B^{2}( \Sigma,\hat{\Sigma} )} \leq W_{2}( \mathbb{P},\mathbb{Q} ),
\end{align}
where $B( \Sigma,\hat{\Sigma} )$ is the Bures distance, and $B^{2}( \Sigma,\hat{\Sigma} ) = \text{tr}( \Sigma + \hat{\Sigma} - 2( \hat{\Sigma}^{\frac{1}{2}} \Sigma \hat{\Sigma}^{\frac{1}{2}} )^{\frac{1}{2}} )$.

Using the Gelbrich bound, upper bounds for the mean and covariance of the worst-case distribution at each time step can be derived.
\newtheorem{lemma}{\bf Lemma} 
\begin{lemma}\label{lemma:Gelbrich_bound}
For the worst-case distribution $\mathbb{P}_{k}^{*}$ at time $k$, its covariance $\mathbb{E}_{\mathbb{P}_{k}^{*}}[ ( {\bar{w}}_{k} - {\hat{\mu}}_{k} ) ( {\bar{w}}_{k} - {\hat{\mu}}_{k} )^{T} | \mathcal{F}_{k - 1}] = {\hat{\Sigma}}_{k}$ and mean $\mathbb{E}_{\mathbb{P}_{k}^{*}}[ {\bar{w}}_{k} | \mathcal{F}_{k - 1}] = {\hat{\mu}}_{k}$ satisfy
\begin{align}\label{eq_mu}
\| {\hat{\mu}}_{k} \| &\leq \textstyle\frac{1}{\sqrt{\lambda_{\min}( C_{s} )}}( \sqrt{\lambda_{\max}( C_{s} ) N} \bar{\mu} + 2\sqrt{2\varepsilon}),\\\label{eq_covariance}
\text{tr}( {\hat{\Sigma}}_{k}) &\leq \textstyle\frac{1}{\lambda_{\min}( C_{s} )}( \sqrt{\lambda_{\max}( C_{s} ) N \text{tr}( \bar{\Sigma} )} + 2\sqrt{2\varepsilon} )^{2},
\end{align}
where $C_s$ is a symmetric positive definite matrix.
\end{lemma}

The proof of Lemma \ref{lemma:Gelbrich_bound} is shown in APPENDIX \ref{app:Gelbrich_bound}. 
{
}
It should be emphasized that the ambiguity set is defined via the Wasserstein distance directly on the disturbance-to-constraint space through $C_s=(F\bar D)^TCF\bar D$, so that the worst‑case distribution automatically focuses on disturbance directions most likely to cause constraint violations. This design achieves adaptive constraint tightening without pre‑specified parameters, greatly enhancing the controller’s flexibility and robustness. It will be naturally embedded into the two‑stage optimization problem formulated next.

\subsection{Two-Stage Distributionally Robust MPC}
To achieve adaptive constraint tightening, a two-stage distributionally robust optimization problem is formulated,
\begin{align}\label{eq_minmax_problem}
\min_{{\bar{u}}_{k} \in \mathcal{U}'}\{ \max_{\mathbb{P}_k \in \mathcal{P}_k}\{ \mathbb{E}_{\mathbb{P}_k}[ V_{q}( {\bar{u}}_{k},{\bar{w}}_{k} ) + V_{c}( {\bar{u}}_{k},{\bar{w}}_{k} ) |\mathcal{F}_{k - 1}] \} \},
\end{align}
where $\mathbb{E}_{{{\mathbb{P}}}_{k}}[ \cdot | \mathcal{F}_{k - 1} ]$ denotes the conditional expectation given $w_0,\cdots,w_{k-1}$, with the random variable distributed as ${\bar{w}}_{k} \sim {{\mathbb{P}}}_{k}$; $\mathcal{U}'=\mathcal{U}^N\cap\mathcal{U}_{k}$
and $\mathcal{U}_{k}$ are defined by \eqref{eq_ternimal_constraints} in Section \ref{sec3}. At every step $k$, the optimal solution of \eqref{eq_minmax_problem} is denoted by ${\bar{u}}_{k}^{*}$ and the real control $u_k={{u}}_{0|k}^{*}$ is exerted.
$V_{q}\left( {\bar{u}}_{k},{\bar{w}}_{k} \right)$ is a quadratic cost function,
\begin{align}\nonumber
   V_{q}( {\bar{u}}_{k},{\bar{w}}_{k} ) &= \textstyle\sum_{i = 0}^{N - 1}{l( x_{i|k},u_{i|k} )} + V_{f}( x_{N|k} ) \\&= \| x_{k} \|_{Q}^{2} + \| {\bar{x}}_{k} \|_{\bar{Q}}^{2} + \| {\bar{u}}_{k} \|_{\bar{R}}^{2} 
\end{align}
with stage cost $l( x_{i|k},u_{i|k} )$ and terminal cost $V_{f}( x_{N|k} )$ given by,
\begin{align*}
    l( x_{i|k},u_{i|k} ) &= x_{i|k}^{T}Q x_{i|k} + u_{i|k}^{T}R u_{i|k}\\
    V_{f}( x_{N|k} ) &= x_{N|k}^{T}P x_{N|k},
\end{align*}
where $Q$, $R$ are weight matrices and $P$ is derived from
\begin{align}\label{eq_Riccati_equation}
A^{T}PA - P + Q - A^{T}PB(R+B^TPB)B^TPA = 0.
\end{align}
$\bar{Q}$ and $\bar{R}$ are block-diagonal weight matrices given by
\begin{align*}
\bar{Q} = \text{diag}(Q,\cdots,Q,P), \quad \bar{R} = \text{diag}(R,\cdots,R).
\end{align*}
$V_{c}( {\bar{u}}_{k},{\bar{w}}_{k} )$ is a penalty term for constraint violation, determined by the second-stage linear program,
\begin{align}\label{eq_second_stage_cost}
V_{c}( {\bar{u}}_{k},{\bar{w}}_{k} ) &= \min_{{y} \geq 0}{[ h^{T}\ 0^{T} ] y}\\\label{eq_second_stage_constraints}
\text{s.t. } D_{1}{y} &= B_{1}{\bar{u}}_{k} + \xi_{k}
\end{align}
where
$$D_{1} = \begin{bmatrix} I & -I \end{bmatrix},\quad B_{1} = F\bar{B},\quad \xi_{k} = F\bar{D}{\bar{w}}_{k} + F\bar{A}x_{k} + G,$$
$h = ( h_{1},\cdots,h_{n_{c}} )^{T}$ are the second-stage cost coefficients of appropriate dimensions, $y$ are the second-stage decision variables, and $\xi_k$ are . The construction is inspired by $L_1$ Exact Penalty Function Method \cite{kerrigan2000soft}, which indicates that sufficiently large penalty weight $h$ can guarantee constraint satisfactory of \eqref{eq_constraints} in the sense of worst-case expectation. Let
$q = F{\bar{x}}_{k} + G = [q_{1},\cdots,q_{n_{c}} ] ^{T}$,
 and introduce auxiliary variables 
$q^{+} =[ q_{1}^{+},\cdots,q_{n_{c}}^{+} ]^{T}, q^{-} = [ q_{1}^{-},\cdots,q_{n_{c}}^{-} ]^{T},$
where $q_{i}^{+} = \max\{ q_{i},0 \}, q_{i}^{-} = \max\{ - q_{i},0 \}.$
It holds that
$q^{+} - q^{-} = q.$
We penalize constraint violations via $q^{+}$, constructing a second-stage optimization problem,
\begin{align*}
&\min_{q^{+},q^{-} \geq 0}{h^{T}q^{+}},\\
\text{s.t. } &q^{+} - q^{-} = F\bar{B}{\bar{u}}_{k} + \xi_{k},
\end{align*}
and \eqref{eq_second_stage_cost}-\eqref{eq_second_stage_constraints} are derived with $y = \lbrack ( q^{+} )^{T},( q^{-} )^{T} \rbrack^{T}$.

This paper aims to handle the case where the disturbance is a non‑stationary process, i.e., its mean and covariance are unknown and time‑varying.
Achieving this goal faces three challenges: obtaining a tractable formulation of the minimax problem \eqref{eq_minmax_problem}, ensuring constraint satisfaction under time‑varying moments, and handling the cross terms induced by a non‑zero worst‑case mean.

\section{Tractable Reformulation of TSDR-MPC}\label{sec3}

This chapter focuses on handling constraints in model predictive control. In MPC, constraints generally fall into two categories: direct constraints, which reflect system performance requirements and physical limits, e.g., \eqref{eq_constraints}, and additional constraints introduced to ensure stability and feasibility, namely terminal ingredients. The former are constraints on future states and must be accurately evaluated under uncertain distributions using suitable methods, whereas the latter are artificially introduced and should be kept as loose as possible to avoid compromising system performance. In the proposed framework \eqref{eq_minmax_problem}, the former are incorporated into the cost function, and adaptive constraint tightening is achieved through the second-stage optimization. As a result, the optimization problem remains feasible at all times, eliminating the need for explicit feasibility constraints, such as terminal invariant sets. For stability-related constraints, only the nominal system is considered, without accounting for unknown distributions. This effectively decouples constraints on future states from terminal constraints applied to the nominal system.

\subsection{Terminal constraints for nominal systems}

In the stability analysis of model predictive control, the construction of a rigorous mathematical framework relies on two core steps: 1) providing an explicit upper bound for the cost function with respect to the optimal solution ${\bar{u}}_{k}^{*}$ at the current time $k$, and 2) establishing a recursive inequality between ${\bar{u}}_{k}^{*}$ and a feasible candidate solution ${\bar{u}}_{k + 1}^{c}$ at time $k+1$, where 
\begin{align}\label{eq_candidate_solution}
{\bar{u}}_{k}^{*}=
\begin{bmatrix}
u_{0|k}^{*}\\\vdots\\u_{N - 1|k}^{*}
\end{bmatrix},{\bar{u}}_{k + 1}^{c} = 
\begin{bmatrix}
u_{0|k + 1}^{c}\\\vdots\\u_{N - 1|k + 1}^{c}
\end{bmatrix}   = 
\begin{bmatrix}
u_{1|k}^{*}\\\vdots\\u_{N - 1|k}^{*}\\u_{f}
\end{bmatrix},
\end{align}
with the terminal control law $u_{f}$.
{
Following the above line of reasoning, this section will first derive an upper-bound expression for the quadratic cost function $V_{q}\left( {\bar{u}}_{k},{\bar{w}}_{k} \right)$ within the TSDR-MPC framework, and then establish a recursive inequality from time $k$ to $k+1$. Building on this, we will design terminal constraints for the nominal system. These constraints do not involve disturbance terms and only tighten the control input via a tuning parameter, thereby providing the necessary foundation for closed-loop stability analysis without affecting the feasibility of online optimization. We first give the terminal control law for the nominal system.

\begin{lemma}\label{lemma:LQR}
Consider the unconstrained nominal system
\begin{align}\label{eq_nominal_dynamics}
z_{i + 1|k} &= Az_{i|k} + Bu_{i|k},\\\label{eq_nominal_initial_state}
z_{0|k} &= x_{k},
\end{align}
under the fixed-gain control law $u_{f} = Kz_{i|k}$ with
\begin{align}\label{eq_Riccati _gain}
K = - ( R + B^{T}PB )^{- 1}B^{T}PA
\end{align}
provided by the Riccati equation \eqref{eq_Riccati_equation}, then the predicted cost satisfies
\begin{align*}
V_{q}( {\bar{u}}_{f},0 ) = \textstyle\sum_{i = 0}^{N - 1}{l( z_{i|k},u_{f} )} + V_{f}( z_{N|k} ) \leq \| x_{k} \|_{P}^{2},
\end{align*}
where ${\bar{u}}_{f} = \lbrack u_{f}^T,\cdots,u_{f}^T\rbrack^T$.
\end{lemma}

This is simply the LQR controller for the unconstrained case. It shows that the cost function $V_{q}( {\bar{u}}_{f},0 )$ obtained using the LQR control law is bounded above by $\| x_{k} \|_{P}^{2}$. 
Moreover, the LQR controller $u_f$ also possesses the descent property
\begin{align}\label{eq_uf_descent_property}
V_{f}( z_{i + 1|k} )= V_{f}( f( z_{i|k},u_{f},0 ) )\leq  V_{f}( z_{i|k} ) - l( z_{i|k},u_{f} ),
\end{align}
which is necessary in the SMPC stability analysis for unbounded noise \cite{11007013}. 
For the terminal cost and stage cost, we have,
\begin{align*}
\| x_{k} \|_{P}^{2} &\leq \lambda_{\max}(P)\| x_{k}\|^{2},\\
l( x_{k},u_{k} ) = \| x_{k} \|_{Q}^{2} &+ \| u_{k} \|_{R}^{2} \geq \| x_{k} \|_{Q}^{2} \geq \lambda_{\min}(Q)\| x_{k} \|^{2}.
\end{align*}
Thus,
\begin{align*}
\| x_{k} \|_{P}^{2} \leq \lambda_{\max}(P)\| x_{k} \|^{2} \leq \textstyle\frac{\lambda_{\max}(P)}{\lambda_{\min}(Q)}l( x_{k},u_{k} ).
\end{align*}
Therefore, for the unconstrained nominal system using the LQR fixed gain, we have
\begin{align*}
V_{q}( {\bar{u}}_{f},0 ) \leq \textstyle\frac{\lambda_{\max}(P)}{\lambda_{\min}(Q)}l( x_{k},u_{k} ).
\end{align*}

On the other hand, when control constraints exist, the nominal MPC typically requires the following assumptions,
\begin{assumption}\label{assumption:admissible_control_sequence}
The sequence $\bar u_f$ is an admissible control sequence for the nominal system \eqref{eq_nominal_dynamics}-\eqref{eq_nominal_initial_state} with input constraints $\bar u\in\mathcal{U}^N$, i.e, $\forall i=1,\cdots,N$, it holds that
\begin{align}\label{eq_positive_invariance}
f(z_{i|k},u_f,0)&\in \mathbb{X}_f, \forall z_{i|k}\in \mathbb{X}_f, \\ \label{eq_control_constraints}
u_f&\in \mathcal{U},
\end{align}
where $\mathbb{X}_f$ is a compact set with the origin in its interior.
\end{assumption}

Conditions \eqref{eq_positive_invariance}-\eqref{eq_control_constraints} are common assumptions in the construction of MPC, where $\mathbb{X}_f$ is supposed to be the maximum positive invariant set to ensure feasibility and constraint satisfactory of states. Following Assumption \ref{assumption:admissible_control_sequence}, we note that in the proposed TSDR-MPC framework, there is no need to explicitly design a terminal invariant set $\mathbb{X}_f$. Unlike conventional robust or stochastic MPC approaches that rely on robust tubes or invariant sets to ensure recursive feasibility, our method embeds constraint satisfaction into the worst-case expectation problem via dual variables and the Wasserstein ambiguity set. Thus, feasibility is maintained as long as $\mathcal{U}'\neq\emptyset$, without introducing additional conservatism from invariant set design. However, for systems subject to disturbances with non‑zero mean, additional care is needed to handle the cross‑terms induced by the non‑zero mean.

We now extend the analysis to the actual TSDR-MPC setting where the optimal control law ${\bar{u}}_{k}^{*}$ is applied in the presence of stochastic disturbances $\bar{w}_k$.
To obtain an upper bound for the cost function in \eqref{eq_minmax_problem}, we need to analyze $V_{q}( {\bar{u}}_{k}^{*},{\bar{w}}_{k} )$. This reduces to considering the quadratic cost $V_{q}( {\bar{u}}_{k}^{*},0 )$ and $V_{q}( {\bar{u}}_{f},0 )$ under distribution $\mathbb{O}$, with corresponding nominal states denoted as $z_{i|k}^{*}$ and $z_{i|k}^{f}$, respectively. An explicit upper bound for the actual quadratic cost $V_{q}( {\bar{u}}_{k}^{*},{\bar{w}}_{k} )$ is provided by the following proposition.

\newtheorem{proposition}{\bf Proposition}
\begin{proposition}\label{proposition:upper_bound_for_quadratic_cost}
Suppose Assumption \ref{assumption_dynamic_bound} hold. For any disturbance sequence ${\bar{w}}_{k}$, the quadratic cost associated with the optimal law ${\bar{u}}_{k}^{*}$ satisfies the following upper bound,
\begin{align*}
V_{q}( {\bar{u}}_{k}^{*},{\bar{w}}_{k} ) \leq 2\textstyle\frac{\lambda_{\max}(P)}{\lambda_{\min}(Q)}l( x_{k},u_{k} ) + g_{1}( {\bar{w}}_{k} ) + c_{1},
\end{align*}
where
\begin{align}\label{eq_g1}
&g_{1}( {\bar{w}}_{k} ) = \| \bar{D}{\bar{w}}_{k} \|_{\bar{Q}}^{2} + 2( \bar{A}x_{k} + \bar{B}{\bar{u}}_{k}^{*} )^{T}\bar{Q}\bar{D}{\bar{w}}_{k},\\\nonumber
c_{1} &= \textstyle2\sum_{i = 0}^{N - 1}\lbrack \lambda_{\max}(Q)( \sum_{j = 1}^{i}{L_{A}^{i - j}L_{B}u_{u}} )^{2} + \lambda_{\max}(R)u_{u}^{2} \rbrack \\&+ \textstyle2\lambda_{\max}(P)( \sum_{j = 1}^{N}{L_{A}^{N - j}L_{B}u_{u}} )^{2}.\nonumber
\end{align}
\end{proposition}

The proof of Proposition  \ref{proposition:upper_bound_for_quadratic_cost} can be seen in APPENDIX \ref{app:upper_bound_for_quadratic_cost}. 
 The upper bound provided by Proposition \ref{proposition:upper_bound_for_quadratic_cost} consists of the current stage cost $l\left( x_{k},u_{k} \right)$, a disturbance‑dependent term $g_{1}\left( {\bar{w}}_{k} \right)$ and a constant term $c_1$ stemming from control deviations. We next establish a recursive inequality between the optimal costs at adjacent time steps.

\begin{proposition}\label{proposition:recurrence_inequality__for_quadratic_cost}
Suppose Assumption \ref{assumption_dynamic_bound} hold. Let ${\bar{u}}_{k + 1}^{c}$ be the candidate control sequence generated by \eqref{eq_candidate_solution} and $x_{i|k + 1}^{c}$ the corresponding state trajectory for $i=0,\cdots,N-1$. The following recursive inequality holds for $\forall\epsilon > 0$,
\begin{align*}
&\textstyle\frac{1}{1 + \epsilon}V_{q}( {\bar{u}}_{k + 1}^{c},{\bar{w}}_{k + 1} ) \leq V_{q}( {\bar{u}}_{k}^{*},0 ) - l( x_{k},u_{k} ) \\&+\textstyle \frac{1}{\epsilon}\sum_{i = 0}^{N - 2}\| \delta_{i} \|_{Q}^{2} + \frac{1}{1 + \epsilon}\lbrack g_{2}( \delta_{N - 1} ) + \delta_{f} \rbrack,
\end{align*}
where
\begin{align}\label{eq_delta_i}
\delta_{i} = &x_{i|k + 1}^{c} - z_{i + 1|k}^{*} ,\\\label{eq_g2}
g_{2}( \delta_{N - 1} ) = V_{f}&( \delta_{N - 1} ) + 2( z_{N|k}^{*} )^{T}P\delta_{N - 1},\\ \label{eq_delta_f}
\delta_{f} = \| Dw_{N - 1|k + 1} \|_{P}^{2} &+ 2( x_{N - 1|k + 1}^{c} )^{T}PDw_{N - 1|k + 1}.
\end{align}
\end{proposition}

The proof of Proposition  \ref{proposition:recurrence_inequality__for_quadratic_cost} can be seen in APPENDIX \ref{app:recurrence_inequality__for_quadratic_cost}. 
Proposition \ref{proposition:recurrence_inequality__for_quadratic_cost} explicitly expresses the error terms induced by disturbances, enabling subsequent moment‑based bounding using \eqref{eq_mu}-\eqref{eq_covariance}. 
However, this recursive inequality involves cross‑terms of the form
$\| z_{N|k}^{*} \|\| {\hat{\mu}}_{k + 1} \|$ in $g_{2}( \delta_{N - 1} )$ and $\delta_{f}$, which arise because the worst‑case mean $\hat \mu_{k+1}$ is generally non‑zero even when the true mean is zero. In general, such cross‑terms can be eliminated by either assuming zero mean or imposing a terminal equality constraint, i.e., $z_{N|k}=0$. Nevertheless, a terminal equality constraint is too restrictive, which severely shrinks the feasible set and degrades control performance. To avoid this drawback, we instead introduce the following terminal constraint  
\begin{align}\label{eq_terminal_constraints}
\| z_{N|k} \|^{2} \leq l_{c}\| x_{k} \|^{2},
\end{align}
where $l_c>0$ is a tuning parameter. 
For the nominal dynamics, the terminal state is $z_{N|k} = A^{N}x_{k} + C_{AB}{\bar{u}}_{k},$ where $C_{AB} = \lbrack A^{N - 1}B\ \cdots\ B \rbrack.$
Hence, the terminal constraints \eqref{eq_terminal_constraints} can be rewritten by $\| A^{N}x_{k} + C_{AB}{\bar{u}}_{k} \|^{2} \leq l_{c}\| x_{k} \|^{2}.$
It can be seen that these quadratic constraints only affect the first-stage feasible set ${\mathcal{U}}^N$ in \eqref{eq_minmax_problem} and can be expressed as the set 
\begin{align}\label{eq_ternimal_constraints}
\mathcal{U}_{k}=\{\bar u_k|\|A^Nx_k+C_{AB}\bar u_k\|^2\leq l_c \|x_k\|^2\}.
\end{align}
Therefore, with $\mathcal{U}^N$ be replaced by $\mathcal{U}'=\mathcal{U}^N\cap\mathcal{U}_{k}$ in \eqref{eq_minmax_problem}, the constraints on future states \eqref{eq_constraints} can be handled via the second-stage optimization that penalizes violations adaptively, without relying on explicit robust tubes.
What remains to be solved is whether $\bar u_f$ is still an admissible control sequence for $\mathcal U'$ and the following proposition provides the answer. 
\begin{proposition}\label{proposition_feasibility}
Suppose Assumption \ref{assumption:admissible_control_sequence} holds, the sequence $\bar u_f$ is an admissible control sequence for TSDR-MPC \eqref{eq_minmax_problem} if $l_c$ is lower bounded by the following inequality,
\begin{align}\label{eq_l_c}
(A_K^N)^TA_K^N\preceq l_cI, ~A_K=A+BK.
\end{align}
\end{proposition}
\begin{proof}
By Assumption \ref{assumption:admissible_control_sequence}, we have $u_f\in\mathcal U$. And if \eqref{eq_l_c} holds, it follows that
$\|z_{N|k}\|^2=x_k^T(A_K^N)^TA_K^Nx_k\leq l_cx_k^Tx_k=l_c\|x_k\|^2$, which indicates that $\bar u_f\in\mathcal{U}'=\mathcal{U}^N\cap\mathcal{U}_{k}$.
\end{proof}

Proposition \ref{proposition_feasibility} guarantees the feasibility of TSDR-MPC \eqref{eq_minmax_problem}. Moreover, by choosing $l_c$ sufficiently large, the terminal constraint can be relaxed arbitrarily, thereby effectively mitigating the performance loss that would otherwise be incurred to restore feasibility.

\subsection{Adaptive constraint tightening}

In this section, we consider constraints on future states, assumed to be of the form \eqref{eq_constraints}. Building upon the recursive inequality established in Proposition 2 and the terminal constraint designed for the nominal system, we now derive a tractable reformulation of the two-stage distributionally robust MPC problem. The following theorem shows that the min-max problem \eqref{eq_minmax_problem} can be transformed into a finite-dimensional optimization problem that is amenable to practical computation. This reformulation exploits strong duality and the structure of the Wasserstein ambiguity set, ultimately yielding a nonconvex program whose solution provides the optimal control sequence and the worst-case expectation.

\newtheorem{theorem}{\bf Theorem}
\begin{theorem}\label{theorem:quivalent_optimization} Suppose that Assumption \ref{assumption:ambiguity set} holds,  and $C_s=(F\bar{D} )^{T}CF\bar{D}$ in \eqref{eq_wasserstein_distance} is symmetric positive definite with some symmetric positive definite matrix $C$. Let $\underline{\gamma} = \rho( ( \frac{1}{2}C_s )^{- 1}{\bar{D}}^{T}\bar{Q}\bar{D} )$, where $\rho(\cdot)$ is the spectral radius. Then the two-stage distributionally robust MPC problem \eqref{eq_minmax_problem} is equivalent to
\begin{align}\nonumber
J(k) =& \min_{\bar u_k \in \mathcal{U}',\gamma > \underline{\gamma}}\big\{ f_1(\bar u_k,\gamma)+ \textstyle\frac{1}{n}\sum_{\sigma=1}^n{\mathcal{V}( \bar u_{k},\gamma,\bar w_{k}^{\sigma} )} \big\}\\& + \| \bar{A}x_{k} \|_{\bar{Q}}^{2},\label{eq_single_min}
\end{align}
where
\begin{align}\label{eq_first_stage_cost}
f_1(\bar u_k,\gamma)=&\| {\bar{u}}_{k} \|_{{\bar{B}}^{T}\bar{Q}\bar{B} + \bar{R}}^{2} + 2( \bar{A}x_{k})^{T}\bar{Q}\bar{B}{\bar{u}}_{k} + \varepsilon\gamma,
\\\nonumber
\mathcal{V}(\bar u_{k},\gamma,\bar w_{k}^{\sigma} ) =& \max_{\pi \in \Pi}\big\{ \textstyle\frac{1}{2}\pi^{T}F\bar{D}C_{1}^{- 1}( F\bar{D} )^{T}\pi \\&+\textstyle \pi^{T}( B_{1}{\bar{u}}_{k} + F\bar{A}x_{k} + G + F\bar{D}C_{1}^{- 1}C_{0} ) \big\}\nonumber\\
& + \textstyle\frac{1}{2}C_{0}^{T}C_{1}^{- 1}C_{0} - \frac{1}{2}\gamma\| {\bar{w}}_{k}^{\sigma} \|_{( F\bar{D} )^{T}CF\bar{D}}^{2},\label{eq_V_max}
\end{align}
with
\begin{align*}
C_{0} &= 2{\bar{D}}^{T}\bar{Q}( \bar{A}x_{k} + \bar{B}{\bar{u}}_{k} ) + \gamma( F\bar{D} )^{T}CF\bar{D}{\bar{w}}_{k}^{\sigma},\\
C_1&=\gamma( F\bar{D} )^{T}CF\bar{D} - 2{\bar{D}}^{T}\bar{Q}\bar{D},
\end{align*}
and $\Pi=\{\pi|D_1^T\pi\leq[h^T 0^T]^T\}$.
\end{theorem}

\begin{proof}
We start from the definition of the cost function in \eqref{eq_minmax_problem},
\begin{align}\label{eq_J_k}
J(k) &= \min_{u \in \mathcal{U}'}\Big\{ \max_{\mathbb{P}_k\in\mathcal{P}_k}\big\{ \left\| {\bar{u}}_{k} \right\|_{\bar{R}}^{2} + \mathbb{E}\big[ \| \bar{A}x_{k} + \bar{B}{\bar{u}}_{k} + \bar{D}{\bar{w}}_{k} \|_{\bar{Q}}^{2} \nonumber\\
&+ V_{c}\left( {\bar{u}}_{k},{\bar{w}}_{k} \right) |\mathcal{F}_{k - 1}\big] \big\} \Big\}\nonumber\\
&= \min_{u \in \mathcal{U}'}\Big\{ \| {\bar{u}}_{k} \|_{{\bar{B}}^{T}{\bar{Q}}_{k}\bar{B} + \bar{R}}^{2} + 2( \bar{A}x_{k} )^{T}{\bar{Q}}_{k}\bar{B}{\bar{u}}_{k}\nonumber\\
& + \max_{\mathbb{P}_k\in\mathcal{P}_k}\big\{ \mathbb{E}\big[ \| \bar w_{k} \|_{{\bar{D}}^{T}\bar{Q}\bar{D}}^{2} + 2( \bar{A}x_{k} + \bar{B}{\bar{u}}_{k} )^{T}\bar{Q}\bar{D}{\bar{w}}_{k} \nonumber\\&+ V_{c}( {\bar{u}}_{k},{\bar{w}}_{k} ) |\mathcal{F}_{k - 1}\big] \big\} \Big\} +\| \bar{A}x_{k} \|_{\bar{Q}}^{2}.
\end{align}
The inner maximization involves only terms that depend on $\bar w_k$. Applying the definition of the Wasserstein ambiguity set \eqref{eq_Wasserstein_ambiguity_set} and using strong duality for the worst-case expectation problem \cite{doi:10.1137/20M1370227}, we obtain
\begin{align*}
&\max_{\mathbb{P}_k\in\mathcal{P}_k}\big\{ \mathbb{E}\big[ \| w_{k} \|_{{\bar{D}}^{T}\bar{Q}\bar{D}}^{2} + 2( \bar{A}x_{k} + \bar{B}{\bar{u}}_{k} )^{T}\bar{Q}\bar{D}{\bar{w}}_{k} \\&+ V_{c}\left( {\bar{u}}_{k},{\bar{w}}_{k} \right) |\mathcal{F}_{k - 1}\big] \big\}= \min_{\gamma\geq0}\big\{\varepsilon\gamma+\textstyle\frac{1}{n}\sum_{\sigma=1}^n\mathcal{V}(\bar u_k,\gamma,\bar w_k^\sigma)\big\},
\end{align*}
where
\begin{align}\label{eq_V}\nonumber
&\mathcal{V}(\bar u_k,\gamma,\bar w_k^\sigma)=\max_{\bar w_k\in\mathcal{W}^{N}}\big\{ - \gamma c( {\bar w}_{k},{\bar w}_{k}^{\sigma} )+\| {\bar{w}}_{k} \|_{{\bar{D}}^{T}\bar{Q}\bar{D}}^{2} \\&+ 2( \bar{A}x_{k} + \bar{B}{\bar{u}}_{k} )^{T}\bar{Q}\bar{D}{\bar{w}}_{k} + V_{c}( {\bar{u}}_{k},{\bar{w}}_{k} )  \big\}
\end{align}

We now focus on solving the inner maximization. Using the dual representation of $ V_{c}\left( {\bar{u}}_{k},{\bar{w}}_{k} \right)$ from \eqref{eq_second_stage_cost}
\begin{align}\label{eq_Vc_dual}
 V_{c}( {\bar{u}}_{k},{\bar{w}}_{k} )=\max_{\pi\in\Pi}{\pi^T(B_1\bar u_k+F\bar D\bar w_k+F\bar Ax_k+G)}.
\end{align}
Substituting \eqref{eq_Vc_dual} into \eqref{eq_V} and interchanging the maximization over $\pi$ and $\bar w_k$, we have

\begin{align}\label{eq_V_max1}
&\mathcal{V}( {\bar{u}}_{k},\gamma,{\bar{w}}_{k}^{\sigma} ) 
= \max_{\pi \in \Pi}\Big\{ \pi^{T}( B_{1}{\bar{u}}_{k} + F\bar{A}x_{k} + G ) + \max_{\bar w_k\in\mathbb{R}^{Nn_w}}\nonumber\\&\big\{ \pi^{T}F\bar{D}{\bar{w}}_{k} +\| {\bar{w}}_{k} \|_{{\bar{D}}^{T}\bar{Q}\bar{D}}^{2} + 2( \bar{A}x_{k} + \bar{B}{\bar{u}}_{k} )^{T}\bar{Q}\bar{D}{\bar{w}}_{k}\nonumber\\
&\textstyle - \frac{\gamma}{2}{{\bar{w}}_{k}}^{T}( F\bar{D} )^{T}CF\bar{D}{\bar{w}}_{k} + \gamma( {\bar{w}}_{k}^{\sigma} )^{T}( F\bar{D} )^{T}CF\bar{D}{\bar{w}}_{k} \big\} \Big\}\nonumber\\
&\textstyle - \frac{\gamma}{2}( {\bar{w}}_{k}^{\sigma} )^{T}( F\bar{D} )^{T}CF\bar{D}{\bar{w}}_{k}^{\sigma}.
\end{align}
The inner maximization over $\bar w_k$ is an unconstrained quadratic program and we rewrite it as
\begin{align*}
&\max_{{\bar{w}}_{k}\in\mathcal{W}^N}\big\{ \| {\bar{w}}_{k} \|_{{\bar{D}}^{T}\bar{Q}\bar{D} - \frac{1}{2}\gamma( F\bar{D})^{T}CF\bar{D}}^{2} + \big[\pi^{T}F\bar{D}\\&+  2( \bar{A}x_{k} + \bar{B}{\bar{u}}_{k} )^{T}\bar{Q}\bar{D}  + \gamma( {\bar{w}}_{k}^{\sigma} )^{T}( F\bar{D} )^{T}CF\bar{D} \big]{\bar{w}}_{k}  \big\}.
\end{align*}

Define
$C_2=C_0+(F\bar D)^T\pi$, the innermost cost function is $- \frac{1}{2}\| {\bar{w}}_{k} \|_{C_{1}}^{2} + C_{2}^{T}{\bar{w}}_{k}$. If $C_1\succ0$ holds for $\gamma>\underline{\gamma}$, then 
\begin{align}\label{eq_worstcase_disturbance}
    \bar w_{k}^{*}= C_{1}^{- 1}C_{2},
\end{align}
and the optimal value of the inner problem is $\frac{1}{2}C_{2}^{T}C_{1}^{- 1}C_{2}.$
Substituting back into \eqref{eq_V_max1} and separating terms involving $\pi$, we obtain \eqref{eq_V_max}.
Finally, substituting back into \eqref{eq_J_k} and combining the minimization over $\gamma$ with the outer minimization over $\bar u_k$, we arrive at the reformulation \eqref{eq_single_min}. The condition $\gamma>\underline{\gamma}$ ensures $C_1\succ0$, which is necessary for the inner maximization to be finite. 
\end{proof}

Theorem \ref{theorem:quivalent_optimization} achieves adaptive constraint tightening by reformulating state constraints through the dual of a second-stage linear program in \eqref{eq_Vc_dual}. 
The key to adaptive constraint tightening lies in the online adjustment of the dual variables $\gamma$ and $\pi$. $\gamma$ is the dual multiplier associated with the Wasserstein radius $\varepsilon$ and controls how far the worst‑case distribution can deviate from the empirical distribution. A larger $\gamma$ reduces the effective ambiguity set and tightens the constraints. $\pi$ is the dual variable of the second‑stage linear program and measures the cost of violating each constraint. Its nonzero components identify the most threatened constraint directions. Because both $\gamma$ and $\pi$ are solved together with the control from the optimization problem, their values automatically increase when the current state approaches a constraint boundary or the uncertainty level rises. This tightens the constraints in real time without any pre‑specified tuning parameters. This state and data dependent adaptation is the essence of adaptive constraint tightening.

In \eqref{eq_second_stage_constraints}, the random vector $\xi_k$ is defined as $\xi_k=F\bar D\bar w_k+F\bar Ax_k+G$, which directly quantifies the potential violation of the predicted state trajectory with respect to the state constraints \eqref{eq_constraints}. The term $F\bar D\bar w_k$ captures the contribution of disturbances, propagated through the system dynamics, to the constraint output. Theorem \ref{theorem:quivalent_optimization} sets the optimal transport cost matrix as $C_s=(F\bar{D} )^{T}CF\bar{D}$, thereby embedding into the Wasserstein distance the same linear mapping $F\bar{D}$ that appears in $\xi_k$. In other words, the distance no longer measures discrepancies between raw disturbance realizations $\bar w_k$, but rather between their images under $F\bar{D}$, that is, their impact on constraint satisfaction.

This design fundamentally ties the geometry of distributional uncertainty to the constraint sensitivity of the system. Since the disturbance dependent part of $\xi_k$ is exactly $F\bar D\bar w_k$,  the transportation penalty $\gamma c_\xi(\xi_k,\xi_k^\sigma)$, where 
\begin{align}\label{eq_c_xi}
c_\xi(\xi_k,\xi_k^\sigma)=\textstyle\frac{1}{2} ( \xi_k-\xi_k^\sigma)^{T} C (\xi_k-\xi_k^\sigma ),
\end{align}
can be interpreted as a cost on the uncertainty of $\xi_k$ itself. 
With $\gamma$ determined adaptively as described above, the penalty on constraint‑relevant uncertainty is strengthened exactly when it is most needed. Consequently, the worst‑case distribution concentrates probability mass on disturbance directions that lead to large $\xi_k$, i.e., severe constraint violations, and the overall constraint tightening emerges naturally from the optimization, driven by the current state and available data, without any pre‑designed robust tube.

In Thereom \ref{theorem:quivalent_optimization}, we assume $C_s$ is positive definite, if not, the entire TSDR‑MPC formulation collapses. The worst‑case disturbance in \eqref{eq_worstcase_disturbance} becomes infinite, rendering the minimax problem \eqref{eq_minmax_problem} ill‑posed and its optimal value undefined. 
Physically, the absence of positive definiteness of $C_s$ means that the Wasserstein transport cost no longer defines a strict metric on the constraint‑output space, making it impossible to penalize harmful uncertainties appropriately. The following lemma gives the necessary and sufficient condition that $C_s$ is positive definite.

\begin{lemma}\label{lemma:turbance_to_constraint_observability}
The matrix $C_s=(F\bar{D})^TCF\bar{D}$ is symmetric and positive definite, if and only if the disturbance‑to‑constraint observability matrix
\begin{align*}
\mathcal{O}_D=
\begin{bmatrix}
    F_0D\\F_0AD\\\vdots\\F_0A^{N-1}D
\end{bmatrix}\in \mathbb{R}^{Nn_c\times n_w}
\end{align*}
has full column rank, i.e., $\text{rank}(\mathcal{O}_D)=n_w.$
\end{lemma}
\begin{proof}
For the fact that $C$ is a symmetric positive definite matrix, $C_s$ is a symmetric positive definite matrix if and only if $F\bar{D}$ has full column rank. Given the block‑diagonal structure $F=\text{diag}(F_0,\cdots,F_0)$ and the lower‑triangular block‑Toeplitz structure of $\bar{D}$, the full column rank of $F\bar{D}$ is equivalent to the full column rank of its first block column, because all subsequent columns are linear combinations of the first block column (shifted in block rows) and the triangular structure ensures injectivity of the whole map if and only if the first block column is injective. The first block column is exactly $\text{rank}(\mathcal{O}_D)$, therefore $\text{rank}(F\bar{D})=Nn_w$ if and only if $\text{rank}(\mathcal{O}_D)=n_w$.
\end{proof}

This necessary and sufficient condition admits that the disturbance must be uniquely identifiable from the sequence of constraint outputs over the prediction horizon. In other words, the mapping from $w_k$ to $(F_0Dw_k,F_0ADw_k,\cdots,F_0A^{N-1}Dw_k)$ must be injective. This is similar to the observability of the pair $(F_0,A)$ over $N$ steps. The full column rank of $\mathcal{O}_D$ guarantees that every direction in the disturbance space generates a unique image in the constraint output space, ensuring that the adaptive constraint tightening mechanism correctly allocates probability mass to the most critical uncertainty directions. Moreover, if $D$ has full column rank and $(F_0,A)$ is observable, $C_s\succ 0$ can be achieved by choosing a sufficiently long prediction horizon $N$.

Although the optimization problem resulting from Theorem \ref{theorem:quivalent_optimization} is nonconvex in general, it can be decomposed into a sequence of tractable linear and quadratic programming subproblems. This decomposition enables efficient online solution via the cutting‑plane algorithm described in the next section, which progressively refines the approximation of the original problem and converges in a finite number of iterations. Hence, despite the nonconvexity, the proposed method remains computationally feasible and well‑suited for real‑time implementation.

\subsection{Cutting-plane algorithm to solve TSDR-MPC}

The cutting plane method is an iterative algorithm widely used in mixed-integer programming, semi-infinite programming, distributionally robust optimization, and other fields\cite{mehrotra2014cutting,kuhn2025distributionally,doi:10.1137/20M1370227}. It works by progressively adding linear constraints, known as cutting planes, to approximate the feasible region of the original problem, thereby transforming a complex problem into a series of easier subproblems. In each iteration, the algorithm solves the current relaxed problem; if the obtained solution violates some constraints of the original problem, a linear inequality is generated that cuts off this infeasible solution while still preserving all feasible solutions. This inequality is then added to the relaxed problem, and the process repeats until convergence is achieved. 

In this section, we adapt the cutting‑plane framework of \cite{doi:10.1137/20M1370227} to the TSDR‑MPC problem. Unlike \cite{doi:10.1137/20M1370227}, where the first‑stage cost is linear, our first‑stage cost is quadratic, which makes the master problem a convex quadratic program. Furthermore, [6] does not provide a finite‑termination proof; we establish such a guarantee in Theorem \ref{theorem:feasibility}. The dual and separation problems, however, retain a structure similar to that in \cite{doi:10.1137/20M1370227}.
Following this adapted framework, we iteratively solve a master problem to obtain the current control sequence and ambiguity set parameter, and then use dual and separation problems to generate new cutting planes (i.e., new dual extreme points and support points) to improve the approximation of the original problem.
At each iteration, the master problem is a lower‑bound approximation of \eqref{eq_single_min} constructed from a collection of support points and dual extreme points. 
Let $\Xi_k'$ be the set of support points. For each $\xi^\omega\in\Xi_k'$, we maintain a set of $\mathcal J^\omega$ previously generated dual extreme points $\pi^j\in \mathcal J^\omega$.
The master problem is then given by
\begin{align}\label{eq_master_problem1}
&\min_{{\bar{u}}_{k} \in \mathcal{U}',\gamma > \underline{\gamma},\nu^{\sigma},\theta^{\omega}}\big\{ f_1(\bar u_k,\gamma) + \textstyle\frac{1}{n}\sum_{\sigma = 1}^{n}\nu^{\sigma} \big\}\\\label{eq_master_problem2}
&\text{s.t. } \nu^{\sigma} \geq \theta^{\omega} - \gamma c_\xi(\xi_k^{\sigma},\xi^{\omega}),\ \forall\sigma = 1,\cdots,n,\ \forall\omega:\xi^{\omega} \in \Xi_k',\\
&\qquad \theta^{\omega} \geq \max_{\pi^j \in \mathcal J^\omega}{\textstyle\big\{ [ \pi^{j} ]^{T}\left( B_{1}{\bar{u}}_{k} +\xi^{\omega} \right) \big\}},\ \forall\omega:\xi^{\omega} \in \Xi_k'.\label{eq_master_problem3}
\end{align}
where the first stage cost $f_1$ is defined by \eqref{eq_first_stage_cost}, $c_\xi$ is defined by \eqref{eq_c_xi}, and $\xi_k^\sigma$ are given sample points from the empirical distribution. 
\eqref{eq_master_problem2} comes from the dual representation of the Wasserstein ambiguity set; \eqref{eq_master_problem3} is an outer linearization of the second‑stage value function in \eqref{eq_Vc_dual}.
Solving master problem yields a candidate $( {\hat{\bar{u}}}_{k},\hat{\gamma},\hat{\theta} )$.

The dual problem refers to the dual problem of the second-stage optimization \eqref{eq_second_stage_cost}-\eqref{eq_second_stage_constraints}, for every $\xi^{\omega}\in\Xi_k'$, we solve
\begin{align} \label{eq_dual_problem1}
V( {\hat{\bar{u}}}_{k},\xi^{\omega} ) &= \max_{\pi \in \Pi}{\pi^{T}( B_{1}{\hat{\bar{u}}}_{k} + \xi^{\omega} )},\\
\Pi &= \{ \pi | 0 \leq \pi \leq h \},\label{eq_dual_problem2}
\end{align}
and obtain the corresponding optimal dual vertex ${\hat{\pi}}^{\omega}$, which is an extreme point of $\Pi$. This problem is a linear program and can be solved efficiently.

The separation problem is given by \eqref{eq_V_max}. For a fixed candidate $(\hat {\bar u}_k,\hat \gamma)$, the separation problem identifies a new scenario $\xi^\omega$ that most violates the current approximation via \eqref{eq_worstcase_disturbance} and the definition of $\xi_k$ in \eqref{eq_second_stage_constraints}, i.e,
\begin{align}\label{eq_xi_omega}
\xi^\omega=F\bar D C_{1}^{- 1}C_{2}+F\bar{A}x_{k} + G.
\end{align}
If $\xi^\omega\notin\Xi_k'$, we add it to the set of support points together with the associated optimal dual vertex ${\hat{\pi}}^{\omega}$. This generates a new cutting plane for the master problem. Algorithm \ref{Algorithm: TSDR_MPC} summarizes the procedure. 

\begin{algorithm}
\caption{TSDR-EMPC}
\label{Algorithm: TSDR_MPC}
\begin{algorithmic}
\WHILE{$k\leq N_k$}
\STATE 1. Obtain $\xi^\sigma_k,\sigma=1,\cdots,n$ and set $\Xi^n_k=\{\xi^1_k,\cdots,\xi^n_k\}$.
    \STATE 2. Initialize $\Xi'_k \leftarrow \Xi^n_k$, iteration $\leftarrow$ True, and for all $\xi^\omega \in \Xi'_k$ let $\mathcal{J}^\omega \leftarrow \emptyset$ and $\theta^\omega \geq 0$.
    \WHILE{iteration = True}
        \STATE 3. Solve master problem \eqref{eq_master_problem1}-\eqref{eq_master_problem3} and obtain $(\hat{\bar{u}}_k, \hat{\gamma})$.
        \FOR{$\xi^\omega \in \Xi'_k$}
            \STATE Solve dual problem \eqref{eq_dual_problem1}-\eqref{eq_dual_problem2} and obtain $\hat{\pi}^\omega$.
            \STATE Append $\hat{\pi}^\omega$ to $\mathcal{J}^\omega$ if $\hat{\theta}^\omega < [\hat{\pi}^\omega]^T(B_1\hat{\bar{u}}_k + \xi^\omega)$.
        \ENDFOR
        \STATE \textbf{if} for any $\xi^\omega \in \Xi'_k$ a new $\hat{\pi}^\omega$ was added to $\mathcal{J}^\omega$ \textbf{then} go to step 3
        \FOR{$\xi^\sigma _k\in \Xi^n_k$}
            \STATE 4. Solve the seperate problem \eqref{eq_V_max} for fixed $(\hat{\bar{u}}_k, \hat{\gamma})$ and obtain $(\hat{\pi}^\omega, \xi^\omega)$ with $\xi^\omega$ given by \eqref{eq_xi_omega}.
            \IF{$\xi^\omega \notin \Xi'_k$} 
                \STATE $\Xi'_k \leftarrow \Xi'_k \cup \{\xi^\omega\}$
                \STATE $\mathcal{J}^\omega \leftarrow \{\hat{\pi}^\omega\}$
            \STATE Add $\theta^\omega \geq \max_{\pi_j \in \mathcal{J}^\omega}\{[\pi^j]^T(B_1\hat{\bar{u}}_k + \xi^\omega)\}$ to \eqref{eq_master_problem2}. 
            \STATE Add $\nu^\sigma \geq \theta^\omega - \gamma \|\xi^\omega - \xi^\sigma_k\|$, $\forall \sigma= 1,\cdots,n$ to \eqref{eq_master_problem3}.
            \ENDIF
        \ENDFOR
        \STATE \textbf{if} no new support points were added \textbf{then} iteration =  False
    \ENDWHILE
    
$\bar{u}_k^* \leftarrow \hat{\bar{u}}_k$

$k\leftarrow k+1$
\ENDWHILE
\end{algorithmic}
\end{algorithm}

\begin{theorem}\label{theorem:feasibility}
Let Assumption \ref{assumption:ambiguity set} hold. Assume that $\text{rank}(\mathcal{O}_D)=n_w$, and $\mathcal U'$ is nonempty at every time step $k$, Algorithm \ref{Algorithm: TSDR_MPC} terminates after a finite number of iterations and returns an optimal solution $\bar u_k^*$ of the original TSDR‑MPC problem \eqref{eq_minmax_problem}. Moreover, all intermediate problems are feasible.
\end{theorem}

\begin{proof}
Since $\mathcal U'\neq  \emptyset$, the original problem \eqref{eq_minmax_problem} admits a feasible solution, and we denote its corresponding cost as $J_0$. Therefore, the master problem \eqref{eq_master_problem1}-\eqref{eq_master_problem3}, which is a relaxation of \eqref{eq_single_min} and by Theorem \ref{theorem:quivalent_optimization} the relaxation of \eqref{eq_minmax_problem}, is also feasible. At every iteration $t$, the cost function of the master problem, which contains the term $\epsilon\gamma$, satisfies $J_t(k)\leq J_0$, ensuring the uniform boundedness of $\gamma^{(t)}$. Moreover, Lemma \ref{lemma:turbance_to_constraint_observability} yields positive definite and bounded $C_1$, and $C_2$ is bounded as long as the current state $x_k$ and the sampling disturbance $\bar w_k^\sigma$ are finite. Thus all worsr-case disturbance $\bar w_k^*$ given by \eqref{eq_worstcase_disturbance} lie in a fixed compact set. Under $\text{rank}(\mathcal{O}_D)=n_w$, the map $F\bar D$ is injective and sends compact sets to compact sets. Hence, every support point $\xi^\omega$ belongs to a fixed compact set $\Xi$.

For the dual problem \eqref{eq_dual_problem1}-\eqref{eq_dual_problem2}, the dual feasible set $\Pi$ is a compact polyhedron with finitely many extreme points. This indicates that the dual problem is feasible and that the set of extreme points $\mathcal J^\omega$ is finite. Also, the compactness of $\Pi$ and the boundedness of $\gamma$ guarantee the feasibility of the separate problem \eqref{eq_V_max}. We then show the finiteness of support points in $\Xi'_k$ by contradiction. Suppose that infinitely many distinct support points $\{\xi^{(t)}\}\in\Xi$ are generated. By compactness, there exists a cluster point $\xi^*$ and a subsequence $\xi^{(t_i)}\to\xi^*$. Each $\xi^{(t_i)}$ is obtained from the separation problem \eqref{eq_V_max} using the current master candidate $(\bar u_k^{(t_i)},\gamma^{(t_i)})$ and yields an optimal dual vertex $\pi^{(t_i)}$. Since there are only finitely many dual vertices, we can extract a further subsequence with constant $\pi^{(t_i)}=\pi^*$. The bounded sequence $(\bar u_k^{(t_i)},\gamma^{(t_i)})$ has a convergent subsequence with limit $(\bar u_k^{*},\gamma^{*})$. Denote the explicit formula of \eqref{eq_xi_omega} as $\xi^{(t_i)}=\Phi_{\pi^*}(\bar u_k^{(t_i)},\gamma^{(t_i)})$, where $\Phi_{\pi^*}$ is a continuous function. Hence $\xi^*=\Phi_{\pi^*}(\bar u_k^*,\gamma^*)$. whenever a new support point $\xi^{(t_i)}$ is added to $\Xi'_k$, the separation problem guarantees that the current master solution violates the associated cutting planes \eqref{eq_master_problem2}-\eqref{eq_master_problem3}. 
Adding an effective cut strictly restricts the feasible set of the convex master problem, so the subsequence $\{J_t(k)\}$ increases strictly. Infinitely many such strict increases contradict boundedness. Therefore, only finitely many distinct support points are ever generated, and Algorithm \ref{Algorithm: TSDR_MPC} will terminate after finite iterations and provide the optimal solution $\bar u_k^*$ for the origin problem \eqref{eq_minmax_problem}.

\end{proof}

The cutting-plane algorithm reformulates the intractable TSDR‑MPC minimax problem into a sequence of tractable master, dual, and separation subproblems. By iteratively adding violated dual vertices and new support points, it constructs progressively tighter lower bounds and detects optimality when no further violating scenarios exist. Theorem \ref{theorem:feasibility} guarantees that the algorithm terminates in finitely many iterations, returns a globally optimal solution, and maintains feasibility throughout. This decomposition enables real-time implementation.

\section{Stability Analysis}\label{sec 4}
This chapter aims to prove the stability and asymptotic performance of the closed-loop system under the proposed TSDR-MPC scheme. To achieve this, we need to analyze the evolution of the closed-loop expected cost over time.
The stability argument hinges on establishing a recursive relationship among the optimal expected costs $J_k\left(\bar{u}_k^\ast, \mathbb{P}_k^\ast\right)$ across consecutive time steps. To analyze the closed-loop stability under the proposed TSDR-MPC scheme, we begin by the following lemma, which characterizes a minimax equilibrium property of the optimal solution pair $\left(\bar{u}_k^\ast, \mathbb{P}_k^\ast\right)$.

\begin{lemma}\label{lemma:minmax}
For the minimax problem \eqref{eq_minmax_problem} defined over control sequences $\bar u$ and distributions $\mathbb{P}$, let $J_k(\bar u_, \mathbb{P})$ denote the expected cost under distribution $\mathbb{P}$ with control $\bar u$, i.e.,
\begin{align}\label{eq_conditional_expected_cost}
J_{k}\left( \bar u,\mathbb{P} \right) = \mathbb{E}_{\mathbb{P}}\left[ V_{q}\left( \bar u,{\bar{w}}_{k} \right) + V_{c}\left( \bar u,{\bar{w}}_{k} \right) | \mathcal{F}_{k - 1} \right],
\end{align}
If $\bar u_k^*$ and $\mathbb{P}_k^*$ are the optimal control sequence and the worst-case distribution at time $k$, respectively, then for any other admissible control $\bar u\in\mathcal{U}'$ and any other distribution $\mathbb{P}\in\mathcal{P}_k$, the following inequalities hold
\begin{align}
J_k\left(\bar{u}_k^\ast, \mathbb{P}\right) \le J_k\left(\bar{u}_k^\ast, \mathbb{P}_k^\ast\right) \le J_k\left(\bar{u}, \mathbb{P}_k^\ast\right).
\end{align}
\end{lemma}

This lemma reveals the “equilibrium” property of the optimal solution $(\bar u_k^*, \mathbb{P}_k^*)$: with the worst-case distribution $\mathbb{P}_k^*$ fixed, the optimal control $\bar u_k^*$ achieves the minimum cost among all admissible controls; and with the optimal control $\bar u_k^*$ fixed, the worst-case distribution $\mathbb{P}_k^*$ leads to the maximum cost among all possible distributions. This property serves as the key theoretical basis for comparing costs under different control sequences and different distributions in the subsequent stability analysis. In particular, if Assumption \ref{assumption:ambiguity set} holds, it directly implies the following two inequalities, which are frequently used in the derivation
\begin{align}\label{eq_minmax1}
J_{k+1}\left(\bar{u}_{k+1}^\ast, \mathbb{P}_{k+1}^\ast\right) &\le J_{k+1}\left(\bar{u}_{k+1}^c, \mathbb{P}_{k+1}^\ast\right) ,\\\label{eq_minmax2}
J_k\left(\bar{u}_k^\ast, \mathbb{O}\right) &\le J_k\left(\bar{u}_k^\ast, \mathbb{P}_k^\ast\right).
\end{align}

A natural derivation path is constructed as follows: First, compare the cost at time $k+1$ under the optimal control and the worst-case distribution $J_{k+1}\left(\bar u_{k+1}^\ast, \mathbb{P}_{k+1}^\ast\right)$ with the cost of a carefully constructed candidate control sequence under the same distribution $J_{k+1}\left(\bar u_{k+1}^c, \mathbb{P}_{k+1}^\ast\right) $. Then, relate the cost of this candidate sequence further to the cost of the optimal control at time $k$ under the zero-disturbance distribution $J_k\left(\bar u_k^\ast, \mathbb{O}\right)$. Finally, connect this nominal cost to the actual optimal cost at time $k$ under the worst-case distribution $J_k\left(\bar u_k^\ast, \mathbb{P}_k^\ast\right)$.
Lemma \ref{lemma:minmax} directly yields two key inequalities \eqref{eq_minmax1}-\eqref{eq_minmax2}. Consequently, the crucial step for the stability proof reduces to establishing a quantifiable relationship between the cost of the candidate solution under the worst-case distribution $J_{k+1}\left(\bar u_{k+1}^c, \mathbb{P}_{k+1}^\ast\right) $ and the nominal optimal cost at the previous time step $J_k\left(\bar u_k^\ast, \mathbb{O}\right)$.

Following this logic and building upon the recursive inequality for the quadratic cost $V_q(\bar u_k^*,\bar w_k)$ derived in Section \ref{sec3}, our subsequent focus shifts to analyzing the behavior of the constraint violation penalty term $V_c(\bar u_k^*,\bar w_k)$. The term $V_c(\bar u_k^*,\bar w_k)$ measures the cost incurred due to constraint violations caused by disturbances, and its expectation directly impacts the total cost $J_k\left(\bar u_k^\ast, \mathbb{P}_k^\ast\right)$. To bound it effectively, we examine two distinct cases based on whether the current state $x_k$ lies within the nominal feasible set. This detailed discussion culminates in Proposition \ref{proposition_constraint_violation_cost}, which provides a unified upper bound for $V_c(\bar u_k^*,\bar w_k)$.

\begin{proposition}\label{proposition_constraint_violation_cost}
Suppose Assumption \ref{assumption_dynamic_bound}-\ref{assumption:ambiguity set} hold.
For any given positive constant $\epsilon_{c1} > 0$, the constraint violation cost $V_c(\bar u_k^*,\bar w_k)$ defined in \eqref{eq_second_stage_cost} satisfies the following inequality,
\begin{align*}
V_{c}( {\bar{u}}_{k}^{*},{\bar{w}}_{k} ) 
&\leq 3\epsilon_{c1}\| h \|^{2} 
+ \textstyle\frac{1}{4\epsilon_{c1}}L_{B_1}^{2}Nu_{u}^{2} 
+ \frac{1}{4\epsilon_{c1}}\| F\bar{D}{\bar{w}}_{k} \|^{2} 
\\&+ \textstyle\frac{1}{4\epsilon_{c1}}\| F\bar{A}Dw_{k - 1} \|^{2}.
\end{align*}
where $\|F\bar B\|\leq L_{B_1}$.
\end{proposition}

\begin{proof}
The proof proceeds by considering two cases, depending on whether the current state $x_k$ lies inside the nominal feasible set.

1) If $x_k$ is inside the nominal feasible set, there exists an admissible control sequence $u_k^z$ that satisfies all constraints and steers the nominal state into the forward invariant terminal set. We relate $V_c(u_k^*, w_k)$ to $V_c(u_k^z, 0)$. By definition,
\begin{align*}
V_{c}( {\bar{u}}_{k}^{*},{\bar{w}}_{k} ) &= \max_{\pi \in \Pi}\{ \pi^{T}( F\bar{B}{\bar{u}}_{k}^{*} + F\bar{D}{\bar{w}}_{k} + F\bar{A}x_{k} + G )\},
\\
V_{c}( {\bar{u}}_{k}^{z},0 ) &= \max_{\pi \in \Pi}\{ \pi^{T}( F\bar{B}{\bar{u}}_{k}^{z} + F\bar{A}x_{k} + G ) \}=0.
\end{align*}
Using properties of the maximum function, we have,
\begin{align*}
&V_{c}( {\bar{u}}_{k}^{*},{\bar{w}}_{k} ) -V_{c}( {\bar{u}}_{k}^{z},0 )
\leq  
 \max_{\pi \in \Pi}\{ \pi^{T}( F\bar{B}( {\bar{u}}_{k}^{*} - {\bar{u}}_{k}^{z} ) + F\bar{D}{\bar{w}}_{k} ) \} \\
&\leq h^{T}\max\{ 0,F\bar{B}\left( {\bar{u}}_{k}^{*} - {\bar{u}}_{k}^{z} \right) \} 
+ h^{T}\max\{ 0,F\bar{D}{\bar{w}}_{k} \} \\
&\leq \| h \|(\| F\bar{B}( {\bar{u}}_{k}^{*} - {\bar{u}}_{k}^{z} ) \| +\| F\bar{D}{\bar{w}}_{k} \|) .
\end{align*}

Applying Young's inequality to the two terms above, and noting that $V_c(u_k^z, 0) = 0$ for the feasible solution, we obtain,
\begin{align*}
V_{c}( {\bar{u}}_{k}^{*},{\bar{w}}_{k} ) \leq 2\epsilon_{c1}\| h \|^{2} 
+ \textstyle\frac{1}{4\epsilon_{c1}}L_{B_1}^{2}Nu_{u}^{2} 
+ \frac{1}{4\epsilon_{c1}}\| F\bar{D}{\bar{w}}_{k} \|^{2}.
\end{align*}

2) If $x_k$ is outside the nominal feasible set, consider a candidate control sequence $\bar u_k^{zc}$ constructed by shifting from a feasible solution at the previous time step $k-1$ as \eqref{eq_candidate_solution}. Its feasibility implies
\begin{align*}
F[ \bar{A}( x_{k} - Dw_{k - 1} ) + \bar{B}{\bar{u}}_{k}^{zc} ] + G \leq 0,
\end{align*}
and hence
$V_c(u_k^{zc}, 0) \leq \|h\| \|F \bar A D w_{k-1}\|$. 
It follows that
\begin{align*}
&V_{c}( {\bar{u}}_{k}^{*},{\bar{w}}_{k} )
= \max_{\pi \in \Pi}\big\{ \pi^{T}( F \bar{A} x_{k} + F\bar{B}{\bar{u}}_{k}^{zc} + G)  \\&+ \pi^{T}(  F\bar{B}({\bar u_k^*-\bar{u}}_{k}^{zc})+F\bar D\bar w_k) \big\} \\
&\leq V_{c}( {\bar{u}}_{k}^{zc},0 ) + \max_{\pi \in \Pi}\{ \pi^{T}(  F\bar{B}({\bar u_k^*-\bar{u}}_{k}^{zc})+F\bar D\bar w_k) \}.
\end{align*}
The first term is exactly $V_c(\bar u_k^{zc}, 0)$, and the second term is handled identically to Case 1. Thus
\begin{align*}
V_{c}( {\bar{u}}_{k}^{*},{\bar{w}}_{k} )\leq V_{c}( {\bar{u}}_{k}^{zc},0  \|(\| F\bar{B}( {\bar{u}}_{k}^{*} - {\bar{u}}_{k}^{zc} ) \| +\| F\bar{D}{\bar{w}}_{k} \|)
\end{align*}
Substituting the bound for $V_c(\bar u_k^{zc}, 0)$ and applying Young's inequality separately to each of the three terms on the right yield that
\begin{align*}
V_{c}( {\bar{u}}_{k}^{*},{\bar{w}}_{k} ) 
\leq& 3\epsilon_{c1}\| h \|^{2} 
+ \textstyle\frac{1}{4\epsilon_{c1}}L_{B_1}^{2}Nu_{u}^{2} 
+ \frac{1}{4\epsilon_{c1}}\| F\bar{D}{\bar{w}}_{k} \|^{2} \\&
+ \textstyle\frac{1}{4\epsilon_{c1}}\| F\bar{A}Dw_{k - 1} \|^{2}. 
\end{align*}
Since the bound in Case 2 is strictly larger than that in Case 1, the more general bound presented in the proposition holds for all cases.
\end{proof}

Building upon the upper bound for the constraint violation cost $V_c(\bar u_k^*,\bar w_k)$ established in Proposition \ref{proposition_constraint_violation_cost}, we now integrate this result with the quadratic cost bound from Proposition \ref{proposition:upper_bound_for_quadratic_cost} to derive an overall upper bound for the optimal expected cost $J_{k}\left( {\bar{u}}_{k}^{*},\mathbb{P}_{k}^{*} \right)$.
Recall the definition of the expected cost from \eqref{eq_conditional_expected_cost}. Applying the bounds from Propositions \ref{proposition:upper_bound_for_quadratic_cost} and \ref{proposition_constraint_violation_cost}, we derive the following upper bound for the optimal expected cost at time $k$,

\begin{align}\nonumber
J_{k}( {\bar{u}}_{k}^{*},\mathbb{P}_{k}^{*} ) 
=& \mathbb{E}_{\mathbb{P}_{k}^{*}}[ V_{q}( {\bar{u}}_{k}^{*},{\bar{w}}_{k} ) + V_{c}\left( {\bar{u}}_{k}^{*},{\bar{w}}_{k} \right) | \mathcal{F}_{k - 1} ] \\\nonumber
\leq& 2\textstyle\frac{\lambda_{\max}(P)}{\lambda_{\min}(Q)}l( x_{k},u_{k} ) + c_{1} + 3\epsilon_{c1}\| h \|^{2} \\&
+ \mathbb{E}_{\mathbb{P}_{k}^{*}}[ g_{1}( {\bar{w}}_{k} ) + \textstyle\frac{1}{4\epsilon_{c1}}\| F\bar{D}{\bar{w}}_{k} \|^{2} | \mathcal{F}_{k - 1} ] \nonumber\\
& 
+ \textstyle\frac{1}{4\epsilon_{c1}}L_{B_1}^{2}Nu_{u}^{2} 
+ \frac{1}{4\epsilon_{c1}}\| F\bar{A}Dw_{k - 1} \|^{2}.\label{eq_J_k_upper_bound}
\end{align}
This bound decomposes the cost into a term proportional to the current stage cost, disturbance-dependent terms and constant offsets. 
To analyze the evolution of the expected cost, we consider the difference between $J_{k+1} (\bar{u}_{k+1}^\ast, \mathbb{P}_{k+1}^\ast) - J_k(\bar{u}_k^\ast, \mathbb{P}_k^\ast)$. We begin by establishing an intermediate inequality using the properties of the minimax problem and the candidate control sequence.
From Lemma \ref{lemma:minmax}, we have
$
J_{k+1} (\bar{u}_{k+1}^\ast, \mathbb{P}_{k+1}^\ast) \leq J_{k+1} (\bar{u}_{k+1}^c, \mathbb{P}_{k+1}^\ast),
$
where $\bar{u}_{k+1}^c$ is the candidate control sequence defined in \eqref{eq_candidate_solution}. This inequality holds because $\bar{u}_{k+1}^*$ is the minimizer for the problem at time $k+1$, whereas $\bar{u}_{k+1}^c$ is a feasible candidate.
Now, we express $J_{k+1} (\bar{u}_{k+1}^c, \mathbb{P}_{k+1}^\ast)$, using its definition and then apply the recursive inequality for the quadratic cost from Proposition \ref{proposition:recurrence_inequality__for_quadratic_cost}. We have
\begin{align*}
&\textstyle\frac{1}{1+\epsilon} J_{k+1}(\bar{u}_{k+1}^\ast, \mathbb{P}_{k+1}^\ast)  \le \frac{1}{1+\epsilon} V_c(\bar{u}_{k+1}^c,\bar{w}_{k+1})\\&\textstyle+V_q\left(\bar{u}_k^\ast, 0\right) - l(x_k, u_k)  + \frac{1}{\epsilon} \sum_{i=0}^{N-2} \mathbb{E}_{\mathbb{P}_{k+1}^\ast} [ \|\delta_i\|_Q^2 | \mathcal{F}_k ] \\
&+ \textstyle\frac{1}{1+\epsilon} \mathbb{E}_{\mathbb{P}_{k+1}^\ast} [ g_2(\delta_{N-1}) + \delta_f ] .
\end{align*}

To further simplify, we bound $V_c(\bar{u}_{k+1}^c\bar{w}_{k+1})$ using a similar approach as in Proposition 3, which yields
\begin{align*}
V_c(\bar{u}_{k+1}^c, \bar{w}_{k+1}) &\le V_c(\bar{u}_k^\ast, 0) + h^T \max \{0, F\bar{D}\bar{w}_{k+1}\} \\
&\textstyle\le V_c(\bar{u}_k^\ast, 0) + \epsilon_{c2} \|h\|^2  + \frac{1}{4\epsilon_{c2}} \|F\bar{D}\bar{w}_{k+1}\|^2,
\end{align*}
for all $\epsilon_{c2} > 0$.
Note that for the zero-distribution $\mathbb{O}$, it follows by Lemma \ref{lemma:minmax} that
$
    J_k(\bar{u}_k^\ast, \mathbb{O}) \le J_k(\bar{u}_k^\ast, \mathbb{P}_k^\ast).
$
Then, substituting this into the previous inequality yields,
\begin{align*}
&\textstyle\frac{1}{1+\epsilon} J_{k+1}(\bar{u}_{k+1}^\ast, \mathbb{P}_{k+1}^\ast) \le J_k(\bar{u}_k^\ast, \mathbb{P}_k^\ast) - l(x_k, u_k) \\
&+ \textstyle\frac{1}{\epsilon} \sum_{i=0}^{N-2} \mathbb{E}_{\mathbb{P}_{k+1}^\ast} [\|\delta_i\|_Q^2 | \mathcal{F}_k] + \frac{1}{1+\epsilon} \mathbb{E}_{\mathbb{P}_{k+1}^\ast} \big[ g_2 (\delta_{N-1}) + \delta_f  \\
& \textstyle + \frac{1}{4\epsilon_{c2}} \|F\bar{D}\bar{w}_{k+1}\|^2 | \mathcal{F}_k \big] + \frac{1}{1+\epsilon} \epsilon_{c2} \|h\|^2.
\end{align*}
Substituting this bound and multiplying both sides by $1+\epsilon$ yield that
\begin{align*}
&J_{k+1} (\bar{u}_{k+1}^\ast, \mathbb{P}_{k+1}^\ast) - J_k(\bar{u}_k^\ast, \mathbb{P}_k^\ast) \le\\
& \textstyle\epsilon J_k(\bar{u}_k^\ast, \mathbb{P}_k^\ast) - (1+\epsilon) l(x_k, u_k)  + \frac{1+\epsilon}{\epsilon} \sum_{i=0}^{N-2} \mathbb{E}_{\mathbb{P}_{k+1}^\ast} [\|\delta_i\|_Q^2 | \mathcal{F}_k] \\&\textstyle
+ \mathbb{E}_{\mathbb{P}_{k+1}^\ast} [ g_2 (\delta_{N-1}) + \delta_f + \frac{1}{4\epsilon_{c2}} \|F\bar{D}\bar{w}_{k+1}\|^2 | \mathcal{F}_k ] + \epsilon_{c2} \|h\|^2
\end{align*}
Substituting the upper bound for $J_k(\bar{u}_k^\ast, \mathbb{P}_k^\ast)$ in \eqref{eq_J_k_upper_bound} yields
\begin{align}\nonumber
J_{k+1} &(\bar{u}_{k+1}^\ast, \mathbb{P}_{k+1}^\ast) - J_k(\bar{u}_k^\ast, \mathbb{P}_k^\ast) \\\nonumber
&\textstyle\le [ 2\epsilon \frac{\lambda_{\max}(P)}{\lambda_{\min}(Q)} - (1+\epsilon) ] l(x_k, u_k) + \epsilon [c_1 + 3\epsilon _{c1} \|h\|^2 \\\nonumber
&\textstyle\quad + \frac{1}{4\epsilon_{c1}} L_{B_1}^2Nu_u^2  + \frac{1}{4\epsilon_{c1}} \|F \bar A D{w}_{k-1}\|^2]+ \epsilon_{c2} \|h\|^2 \\\nonumber
&\textstyle\quad + \frac{1+\epsilon}{\epsilon} \sum_{i=0}^{N-2} \mathbb{E}_{\mathbb{P}_{k+1}^\ast} [\|\delta_i\|_Q^2 | \mathcal{F}_k] \\\nonumber
&\textstyle\quad + \mathbb{E}_{\mathbb{P}_{k+1}^\ast} [ g_2 (\delta_{N-1}) + \delta_f + \frac{1}{4\epsilon_{c2}} \|F\bar{D}\bar{w}_{k+1}\|^2 | \mathcal{F}_k ] \\
&\textstyle\quad + \mathbb{E}_{\mathbb{P}_k^\ast} [ g_1 (\bar w_k) + \frac{1}{4\epsilon_{c1}} \|F\bar{D}\bar{w}_k\|^2 | \mathcal{F}_{k-1} ]. \label{eq_cost_difference_inequality}
\end{align}

The above inequality contains several conditional expectation terms that depend on the disturbance sequences and their interaction with the system state. To facilitate the stability analysis, we must bound these terms collectively by expressions involving the current state norm $\|x_k\|^2$ and the statistical measures of the disturbances. The culmination of this term-by-term analysis is the following Proposition, which provides an upper bound for the entire collection of conditional expectation terms.

\begin{proposition}\label{propositopn:conditional_expectation_term_bound}
 Suppose all Assumption \ref{assumption_dynamic_bound}-\ref{assumption:admissible_control_sequence} hold. There exist positive  constants $k_0$, $k_1$, $k_2$, $k_{31}$, $k_{32}$, $k_4$ and $k_5$ such that the following inequality holds,
\begin{align}\nonumber
&\mathbb{E}_{\mathbb{P}_k^\ast} [g_1 (\bar w_k) + \textstyle\frac{1}{4\epsilon_{c1}} \|F\bar{D}\bar{w}_k\|^2 | \mathcal{F}_{k-1}] +\sum_{i=0}^{N-2} \mathbb{E}_{\mathbb{P}_{k+1}^\ast} [\|\delta_i\|_Q^2 | \mathcal{F}_k] \\&+ \mathbb{E}_{\mathbb{P}_{k+1}^\ast} [g_2 (\delta_{N-1}) + \delta_f + \textstyle\frac{1}{4\epsilon_{c2}} \|F\bar{D}\bar{w}_{k+1}\|^2 | \mathcal{F}_k] \nonumber\\
&\le k_0 \|x_k\|^2 + k_1 \|w_k\|^2 + k_2 \|\hat \mu_{k+1}\|^2 + k_{31} \|\hat \mu_k\| + k_{32} \|\hat \mu_k\|^2 \nonumber\\
&+ k_4 \text{tr}(\hat \Sigma_{k+1}) + k_5 \text{tr}(\hat \Sigma_k),\label{eq_disturbed_terms}
\end{align}
for all positive constants $\epsilon_{c1}, \epsilon_{c2}>0$.
\end{proposition}

It can be seen in the proof (APPENDIX \ref{app:conditional_expectation_term_bound}) that the role of the terminal constraint \eqref{eq_terminal_constraints} is to prevent the emergence of constant terms along with $\|x_k\|^2$ which cannot be eliminated.
Building upon this, we will substitute and simplify the expectation terms in the previously derived cost difference inequality. This step will directly link the performance evolution to the ambiguity set radius $\varepsilon$ and the prior bounds on disturbance moments $\bar\mu$ and $\bar \Sigma$, thereby culminating in the proof that the average closed-loop cost possesses the upper bound described in the following Theorem \ref{theorem_asymptotic_performance_bound}.

\begin{theorem}\label{theorem_asymptotic_performance_bound}
Suppose all Assumption \ref{assumption_dynamic_bound}-\ref{assumption:admissible_control_sequence} hold. There exist functions $\bar{\sigma}_1, \bar{\sigma}_2, \bar{\sigma}_3 \in \mathcal{K}_\infty$ such that the closed-loop system  under the proposed TSDR-MPC scheme satisfies the following asymptotic performance bound,
\begin{equation}\label{eq_asymptotic_performance_bound}
\limsup_{\bar{N} \to \infty} \mathbb{E} \Big[ \frac{1}{\bar{N}} \sum_{k=0}^{\bar{N}} l(x_k, u_k) \Big] \le \bar{\sigma}_1(\varepsilon) + \bar{\sigma}_2(\bar{\mu}) + \bar{\sigma}_3(\text{tr}(\bar{\Sigma})).
\end{equation}
\end{theorem}

\begin{proof}
Let $c_l = 2\frac{\lambda_{\max}(P)}{\lambda_{\min}(Q)} > 1$. Setting $\epsilon \in (0, 1)$, $\epsilon_{c1} =1 $, and $\epsilon_{c2} = \epsilon$, it follows that
$
\epsilon \le \frac{1+\epsilon}{\epsilon}, \quad 1 \le \frac{1+\epsilon}{\epsilon}
$. Substituting above inequalities and \eqref{eq_disturbed_terms} into the cost difference inequality \eqref{eq_cost_difference_inequality} yields
\begin{align*}
&J_{k+1}(\bar{u}_{k+1}^\ast, \mathbb{P}_{k+1}^\ast) - J_k(\bar{u}_k^\ast, \mathbb{P}_k^\ast) \le -( \textstyle\frac{3}{4} - c_l \epsilon ) l(x_k, u_k) + \epsilon c_2 \\
&+ \textstyle\frac{1+\epsilon}{\epsilon} \big(  k_2 \|\hat\mu_{k+1}\|^2 + k_{31} \|\hat\mu_k\|  + k_{32} \|\hat\mu_k\|^2+ k_4 \text{tr}(\hat\Sigma_{k+1}) \\&+ k_5 \text{tr}(\hat\Sigma_k) \big)+\textstyle\frac{\epsilon}{4} \|F \bar A {D} {w}_{k-1}\|^2+ \frac{1+\epsilon}{\epsilon}k_1\|w_k\|^2,
\end{align*}
where
$c_2=c_1 + 4\|h\|^2 + \frac{1}{4} L_{B_1}^2 Nu_u^2$.
Note that
\begin{align*}
\mathbb{E} [\|w_k\|^2] &\le \bar{\mu}^2 + \text{tr}(\bar{\Sigma}), \\
\mathbb{E} [\|F \bar A {D} {w}_{k-1}\|^2] &\le C_{w2} (\bar{\mu}^2 + \text{tr}(\bar{\Sigma})),
\end{align*}
where $ C_{w2}=\lambda_{\text{max}}(F_0^TF_0)L_D^2\sum_{i=1}^NL_A^{2i}$.
Now, we use the moment bounds from Lemma \ref{lemma:Gelbrich_bound} where the worst-case mean and covariance are bounded by expressions involving $\bar \mu, \bar \Sigma$ and $\varepsilon$. Let $c_{sN} = \sqrt{\frac{\lambda_{\max}(C_s)}{\lambda_{\min}(C_s)}N}$, by taking total expectation and applying these moment bounds, we obtain
\begin{align*}
&\mathbb{E}  [ J_{k+1}(u_{k+1}^\ast, \mathbb{P}_{k+1}^\ast) - J_k(u_k^\ast, \mathbb{P}_k^\ast) + \textstyle\frac{1}{2} l(x_k, u_k) ] \\
&\le \mathbb{E} [ -( \textstyle\frac{1}{4} - c_l \epsilon ) l(x_k, u_k) ] + \epsilon c_2 \\
& + ( \textstyle\frac{\epsilon}{4} C_{w2} + \frac{1+\epsilon}{\epsilon} k_1) \bar{\mu}^2 + \frac{1+\epsilon}{\epsilon} k_{31} \frac{\sqrt{\lambda_{\max}(C_s)N} \bar{\mu} + 2\sqrt{2\varepsilon}}{\sqrt{\lambda_{\min}(C_s)}} \\
& + [   \textstyle\frac{1+\epsilon}{\epsilon} (k_2 + k_{32}) ] \frac{(\sqrt{\lambda_{\max}(C_s)N} \bar{\mu} + 2\sqrt{2\varepsilon})^2}{\lambda_{\min}(C_s)} \\
& + ( \textstyle\frac{\epsilon}{4} C_{w2} + \frac{1+\epsilon}{\epsilon} k_1 ) \text{tr}(\bar{\Sigma})  + [\frac{1+\epsilon}{\epsilon} (k_4 + k_5) ]\\
&\times \textstyle\frac{(\sqrt{\lambda_{\max}(C_s)N \text{tr}(\bar{\Sigma})} + 2\sqrt{2\varepsilon})^2}{\lambda_{\min}(C_s)} \\
&\le \mathbb{E} [ -( \textstyle\frac{1}{4} - c_l \epsilon ) l(x_k, u_k) ] + \epsilon c_2 + \sigma_1(\varepsilon) + \sigma_2(\bar{\mu}) + \sigma_3(\text{tr}(\bar{\Sigma}))
\end{align*}
where
\begin{align*}
\sigma_1(\varepsilon) &= [\textstyle\frac{1+\epsilon}{\epsilon} (k_2 + k_{32} + k_4 + k_5) ] \frac{16\varepsilon}{\lambda_{\min}(C_s)}  \\
&\quad+ \textstyle\frac{1+\epsilon}{\epsilon} k_{31} \frac{2\sqrt{2\varepsilon}}{\sqrt{\lambda_{\min}(C_s)}} \\
\sigma_2(\bar{\mu}) &= \textstyle\frac{1+\epsilon}{\epsilon} k_{31} c_{sN} \bar{\mu}+\frac{\epsilon}{4} C_{w2}\bar{\mu}^2  \\
&\quad+\textstyle\frac{1+\epsilon}{\epsilon} (k_1 + 2(k_2 + k_{32}) c_{sN}^2)  \bar{\mu}^2 \\
\sigma_3(\text{tr}(\bar{\Sigma})) &= \textstyle\frac{\epsilon}{4}C_{w2} \text{tr}(\bar{\Sigma})+\frac{1+\epsilon}{\epsilon} (k_1 + 2(k_4 + k_{5}) c_{sN}^2)   \text{tr}(\bar{\Sigma}).
\end{align*}

Now choose $\epsilon = \min \{ \frac{1}{4c_l},\max\{\varepsilon^{1/3}, \bar{\mu}^{1/2}, \text{tr}(\bar{\Sigma})^{1/2}\}  \} \in (0, 1)$, we have $\frac{1+\epsilon}{\epsilon} \le \frac{2}{\epsilon}$. The functions are further bounded as,
\begin{align*}
\sigma_1(\varepsilon) &\le \textstyle\frac{32(k_2 + k_{32} + k_4 + k_5)\varepsilon}{\lambda_{\min}(C_s)\epsilon}+ \frac{2}{\epsilon} k_{31} \frac{2\sqrt{2\varepsilon}}{\sqrt{\lambda_{\min}(C_s)}} \\
&\le  \textstyle\frac{32(k_2 + k_{32} + k_4 + k_5)}{\lambda_{\min}(C_s)} \max\{\varepsilon^{2/3},4c_l\varepsilon\}  \\
&\quad+ \textstyle\frac{4\sqrt{2} k_{31}}{\sqrt{\lambda_{\min}(C_s)}} \max\{\varepsilon^{1/6},4c_l\varepsilon^{1/2}\}\\
& = 0.5\bar{\sigma}_1(\varepsilon)-c_2\varepsilon^{1/2}, \\
\sigma_2(\bar{\mu}) &\le \textstyle\frac{2}{\epsilon} k_{31} c_{sN} \bar{\mu}+  [\frac{1}{4} C_{w2}+\frac{2}{\epsilon} (k_1 + 2(k_2 + k_{32}) c_{sN}^2) ]\bar{\mu}^2  \\
&\le \textstyle\frac{C_{w2}}{4}  \bar{\mu}^2 + 2 (k_1 + 2(k_2 + k_{32}) c_{sN}^2) \max\{\bar{\mu}^{\frac{3}{2}} ,4c_l\bar{\mu}^{2} \}\\
&\quad+  2 k_{31} c_{sN} \max\{\bar{\mu}^{1/2} ,4c_l\bar{\mu} \} = 0.5\bar{\sigma}_2(\bar{\mu})-c_2\bar \mu^{1/2}, \\
\sigma_3(\text{tr}(\bar{\Sigma})) &\le 
[ \textstyle\frac{1}{4} C_{w2}+ \frac{2}{\epsilon} k_1 + \frac{4}{\epsilon} (k_4 + k_5) c_{sN}^2 ] \text{tr}(\bar{\Sigma}) \\
&\le  [2 k_1+4 (k_4 + k_5) c_{sN}^2 ]\max\{\text{tr}(\bar{\Sigma})^{1/2},4c_l\text{tr}(\bar{\Sigma})\}    \\
&\quad  +\textstyle\frac{1}{4} C_{w2} \text{tr}(\bar{\Sigma}) = 0.5\bar{\sigma}_3(\text{tr}(\bar{\Sigma}))-c_2\text{tr}(\bar{\Sigma})^{\frac{1}{3}}.
\end{align*}

Finally, after summing from $k=0$ to $\bar{N}-1$ and taking the limit, we obtain \eqref{eq_asymptotic_performance_bound} and completes the proof.
\end{proof}

\begin{remark}
The proposed TSDR-MPC framework possesses a natural degeneracy property, seamlessly connecting several important special cases.
When the ambiguity set radius $\varepsilon = 0$ and the disturbance mean bound $\bar{\mu} = 0$, the Wasserstein ambiguity set collapses to a singleton centered at the empirical distribution. The problem then reduces to sample average approximation (SAA) based stochastic MPC. By choosing $\epsilon = \min \{ \frac{1}{4c_l}, \text{tr}(\bar{\Sigma})^{1/2} \} $, the performance bound in Theorem 2 simplifies to
    \begin{equation*}
    \limsup_{\bar{N} \to \infty} \mathbb{E} \Big[ \frac{1}{\bar{N}} \sum_{k=0}^{\bar{N}-1} l(x_k, u_k) \Big] \le \bar{\sigma}_3(\text{tr}(\bar{\Sigma}))
    \end{equation*} 
When additionally $\bar{\Sigma} = 0$, the system reduces to a disturbance-free deterministic system. Choosing $\epsilon = \frac{1}{4c_l}$, then leads to a zero upper bound for the average cost in Theorem \ref{theorem_asymptotic_performance_bound}, which implies asymptotic convergence of the state to the origin. This is fully consistent with the asymptotic stability conclusion of classical deterministic MPC.

This degeneracy consistency validates the theoretical soundness of the proposed framework: as a general method for handling distributional uncertainty, it naturally converges to classical and simpler cases when the uncertainty diminishes.
\end{remark}

\section{Simulation}\label{sec 5}
To validate the effectiveness, adaptability, and stability of the proposed TSDR-MPC framework, we conduct numerical simulations on the following system commonly adopted in robust and stochastic MPC studies \cite{mayne2005robust,10383526}. We consider the discrete-time linear time-invariant system
\begin{equation}
x_{k+1} = \begin{bmatrix} 1 & 1 \\ 0 & 1 \end{bmatrix} x_k + \begin{bmatrix} 0.5 \\ 1 \end{bmatrix} u_k + \begin{bmatrix} 1 & 0 \\ 0 & 1 \end{bmatrix} w_k,
\end{equation}
where $x_k\in\mathbb R^2$ denotes the state, $u_k\in \mathbb R$ is the control input, and $w_k\in\mathbb R^2$ is a stochastic disturbance with unknown distribution. The initial state is fixed as $x_0=[-5,-2]^T$.

A pre-stabilizing LQR gain $K_0=[-0.2068~-0.6756]$ with $Q_0=0.1I_2,R_0=1$, is introduced to enlarge the feasible region of the first-stage decision variables.
The true control applied to the system is $u_k=K_0x_k+v_k$, where $v_k$ is the optimization variable obtained by solving the TSDR-MPC.

The stage cost and terminal cost matrices are chosen as
\begin{align*}
Q=I_2,R=0.1,P=\begin{bmatrix}
    2.0599 \quad0.5916\\
    0.5916 \quad 1.4228
\end{bmatrix},
\end{align*}
where $P$ is obtained from the solution of the discrete-time algebraic Riccati equation \eqref{eq_Riccati_equation}. The corresponding LQR gain is K=[-0.6167\quad -1.2703].
The second-stage penalty coefficients are chosen as $h=10^3\times \mathbf{1}_{n_c}$ to ensure exact penalty properties. The prediction horizon is set to $N=3$, and the terminal constraint tuning parameter is chosen as $l_c=2$.
The state constraints are defined by 
\begin{equation*}
F_0 = \begin{bmatrix} 1 & 0 \\ -1 & 0 \\ 0 & 1 \\ 0 & -1 \end{bmatrix}, \quad G_0 = \begin{bmatrix} -2 \\ -10 \\ -2 \\ -2 \end{bmatrix},
\end{equation*}
The input constraint is
$\mathcal U=\{u_k|-1\leq u_k\leq1\}.$
For the Wasserstein ambiguity set, we set the radius $\varepsilon=0.01$ and the weighting matrix $C=I_{Nn_c}$. The number of samples for the empirical distribution is set to $n=10$.  
Note that to demonstrate adaptiveness of the proposed method, all the above simulation settings remain unchanged across the different uncertainty scenarios.

\begin{figure}[tbp]
\centerline{\includegraphics[width=1.1\columnwidth]{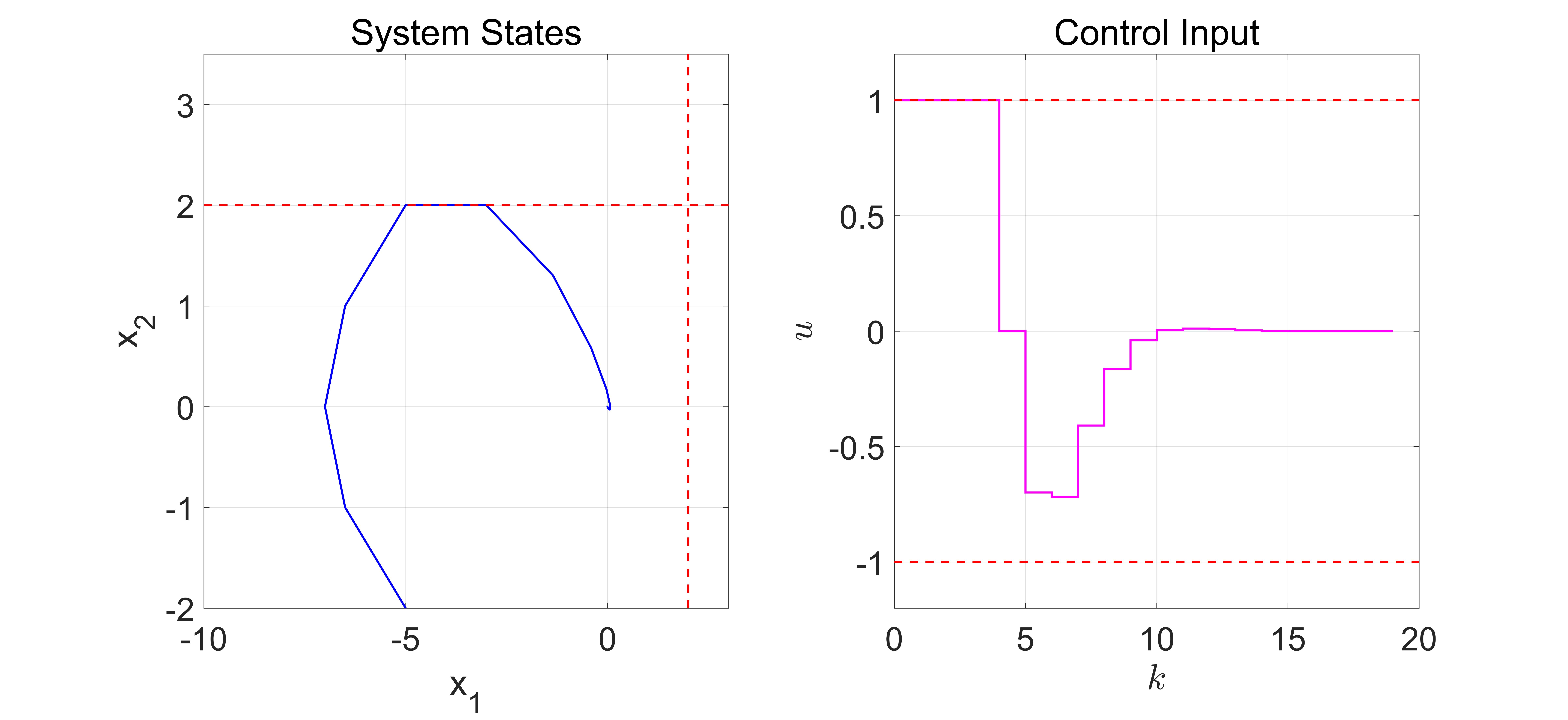}}
\caption{Trajectories of the states (left) and control inputs (right) without disturbance: 
The red dashed lines are constraints.
The simulation example was run for 20 times.}
\label{fig_nominal}
\end{figure}

Fig. \ref{fig_nominal} shows the state trajectories and control inputs under nominal (disturbance‑free) conditions, repeated for 20 simulation runs. All trajectories start from $x_0=[-5\quad -2]^T$ and converge smoothly to the origin while respecting the state constraints. The consistency across all 20 runs demonstrates the deterministic stability of the TSDR‑MPC controller when no disturbances are present.

\begin{figure}[tbp]
\centerline{\includegraphics[width=1.05\columnwidth]{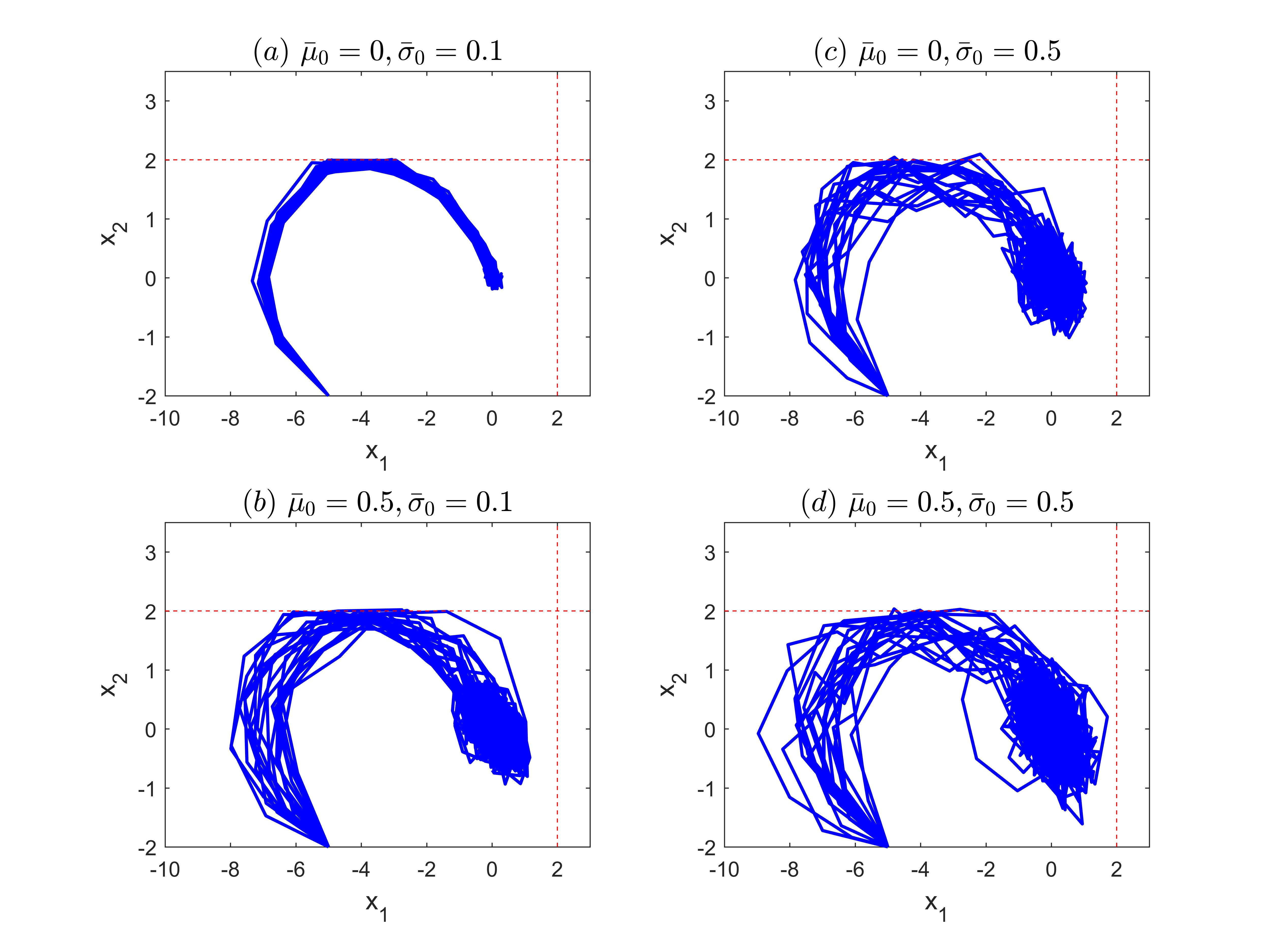}}
\caption{
Closed-loop states under different $\bar \mu_0$ and $\bar \sigma_0$: 
The red dashed lines are constraints.
Each example was run for 20 times.}
\label{fig_dr}
\end{figure}

Fig. \ref{fig_dr} presents the state trajectories under different disturbance moment bounds. We consider four representative scenarios. Here $\bar \mu=\sqrt{2}\mu_0$ and $\bar \Sigma=\text{diag}(\bar \sigma_0,\bar \sigma_0)$.
For each scenario, we first randomly generate a true mean vector $\mu_k$ uniformly from $[-\bar \mu_0,\bar \mu_0]^2$ and a true covariance matrix $\Sigma_k$ whose eigenvalues are uniformly distributed in $[0,\bar \sigma_0^2]$. Then, at each time step $k$, the disturbance $w_k$ is drawn from the Gaussian distribution $\mathcal N(\mu_k,\Sigma_k)$.
This procedure mimics the situation where the true distribution is unknown but its moments are known to lie within given bounds. All other simulation parameters are kept identical across the four scenarios. Each scenario is simulated for 20 independent runs.

In scenario (a) $(\bar\mu_0=0,\bar\sigma_0=0.1)$, trajectories closely follow the nominal path and stay well inside the constraints, showing that the controller does not over‑react to negligible uncertainty.
In scenario (b) $(\bar\mu_0=0.5,\bar\sigma_0=0.1)$, despite a persistent bias, the adaptive tightening automatically adjusts the second‑stage dual variables to counteract the non‑zero mean.
In scenario (c) $(\bar\mu_0=0,\bar\sigma_0=0.5)$, large covariance leads to occasional constraint violations. A larger penalty coefficient $h$ can effectively reduce the violation rate, yet at the cost of increased average stage cost.
In scenario (d) $(\bar\mu_0=0.5,\bar\sigma_0=0.5)$, the most challenging case, the system still converges to a neighbourhood of the origin, confirming the theoretical asymptotic bound of Theorem 3, which grows as the upper bounds of the disturbance moments increase.

\begin{figure}[tbp]
\centerline{\includegraphics[width=1.05\columnwidth]{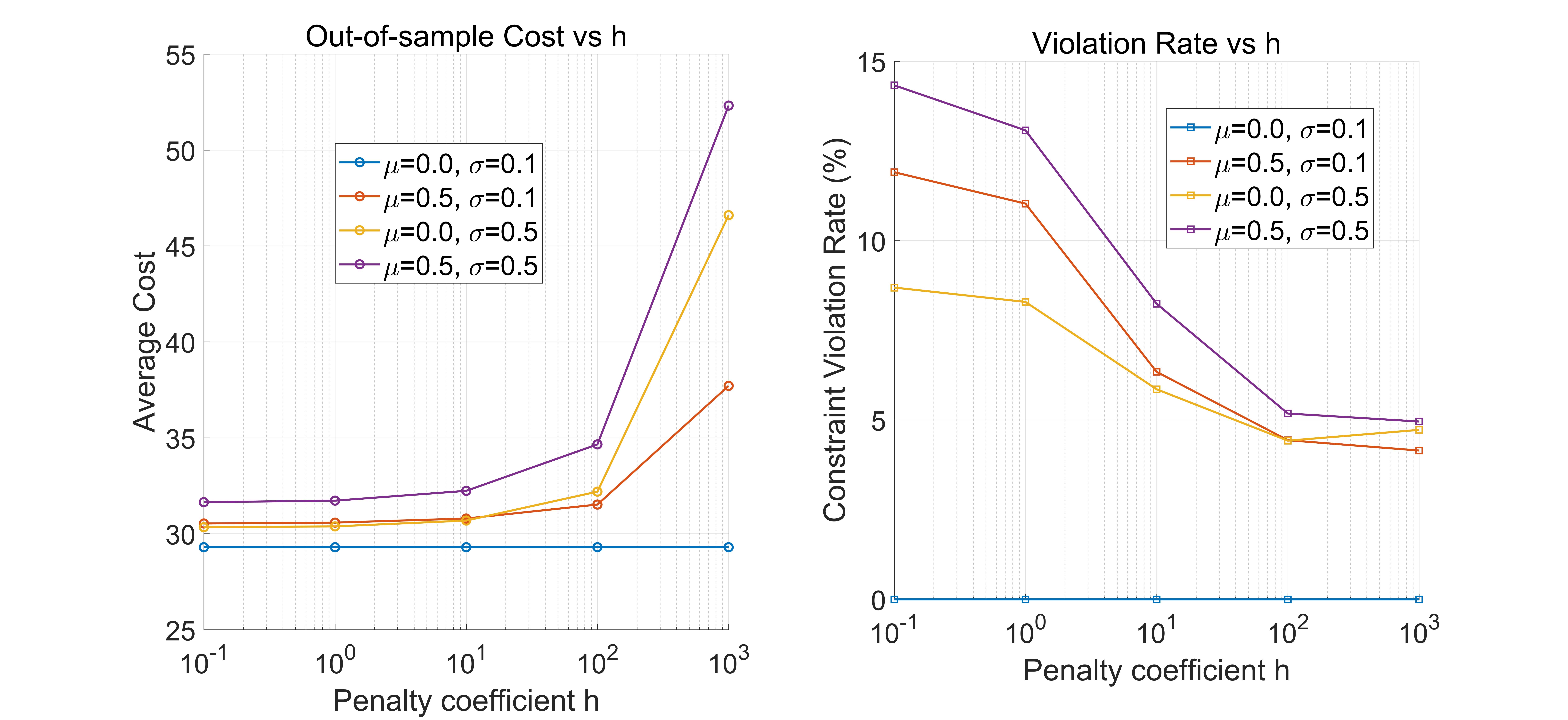}}
\caption{Out‑of‑sample performance versus penalty coefficient.}
\label{fig_out_of_sample}
\end{figure}

Fig \ref{fig_out_of_sample} fixes the initial state at $[-4,~1.98]^T$ and varies $h$ in out‑of‑sample tests. 200 different sampled disturbances are used; each yields a control sequence evaluated on 1000 outer samples, and the results are averaged. It shows that a properly chosen $h$ lowers the violation frequency at the expense of average performance.
These simulation results collectively validate the theoretical claims that the TSDR‑MPC framework achieves adaptive constraint tightening, maintains recursive feasibility through the two‑stage structure, and provides closed‑loop stability with performance bounds despite distributional uncertainty.

\section{Conclusion}\label{sec6}
This paper presents a TSDR-MPC framework that achieves adaptive constraint tightening against disturbances with unknown time-varying moments by embedding violation penalties into a second-stage optimization problem. Leveraging the Wasserstein ambiguity set and duality, the problem is rendered tractable and solved via a finitely convergent cutting-plane algorithm suitable for real-time implementation. Simulation results validate TSDR-MPC's ability to autonomously adjust conservatism based on sampled disturbances, confirming its efficacy in balancing robustness and performance against unknown time-varying disturbances.

\appendix
\subsection{Proof of Lemma \ref{lemma:Gelbrich_bound}}\label{app:Gelbrich_bound}
\begin{proof}
Consider the transformed disturbance sequence ${\bar{w}}_{k}'  =  C_{s}^{1/2} {\bar{w}}_{k}$, where $C_{s}^{1/2}$ is the positive-definite square root of $C_{s}$. Under this transformation, the Wasserstein distance transforms as
\begin{align*}
W_{C}( \mathbb{P},\mathbb{Q} ) =\textstyle\frac{1}{2}W_{2}^2( \mathbb{P}',\mathbb{Q}' ),  
\end{align*}
where $\mathbb{P}'$, $\mathbb{Q}'$ are the corresponding distributions with respect to $\bar w_k'$. From the ambiguity set definition $W_{C}( \mathbb{P}_k,\mathbb{Q}_k )\leq\varepsilon$, we obtain $W_{2}( \mathbb{P}_{k}',\mathbb{Q}_{k}' )\leq\sqrt{2\varepsilon}$. Similarly, for the worst-case distribution $\mathbb{P}_k^*$, we have $W_{2}( \mathbb{P}_{k}^{*'},\mathbb{Q}_{k}' )\leq\sqrt{2\varepsilon}$.  Applying the triangle inequality for the Wasserstein distance yields
\begin{align*}
W_{2}( \mathbb{P}_{k}',\mathbb{P}_{k}^{*'} ) \leq W_{2}( \mathbb{P}_{k}',\mathbb{Q}_{k}' ) + W_{2}( \mathbb{P}_{k}^{*'},\mathbb{Q}_{k}' ) \leq 2\sqrt{2\varepsilon}.
\end{align*}
Now, apply the Gelbrich bound to the standard 2-Wasserstein distance. This gives
\begin{align*}
\| \mu_{k}' - {\hat{\mu}}_{k}' \| &\leq W_{2}( \mathbb{P}_{k}',\mathbb{P}_{k}^{*'} ) \leq 2\sqrt{2\varepsilon},\\
B( \Sigma_{k}',{\hat{\Sigma}}_{k}' )&\leq W_{2}( \mathbb{P}_{k}',\mathbb{P}_{k}^{*'} ) \leq 2\sqrt{2\varepsilon}.
\end{align*}

To bound $\|\mu_k'\|$, note that the true mean vector $\mu_k$ consists of $N$ independent components, each satisfying $\|\mu_k^{(i)}\|\leq\bar \mu$. Hence, $\|\mu_k\|\leq\sqrt{\sum_{i=0}^{N-1}\|\mu_k^{(i)}\|^2}\leq\sqrt{N}\bar \mu$.
Since $\mu_k'=C_s\textstyle^{1/2}\mu_k$, we have
$
\| \mu_{k}' \| = \| C_{s}\textstyle^{1/2}\mu_{k} \| \leq \| C_{s}\textstyle^{1/2} \|\| \mu_{k} \| \leq \sqrt{\lambda_{\max}(  C_{s} )N} \bar{\mu} .
$
Combining this with the bound on $\| {\hat{\mu}}_{k}' -\mu_{k}' \|$ gives
\begin{align*}
\| {\hat{\mu}}_{k}' \| \leq \| \mu_{k}' \| + \| {\hat{\mu}}_{k}' -\mu_{k}' \| = \sqrt{\lambda_{\max}( C_{s} ) N} \bar{\mu} + 2\sqrt{2\varepsilon}.
\end{align*}

Transforming back to the original space, $\hat \mu_k=C_s\textstyle^{-1/2}{\hat{\mu}}_{k}'$, so
\begin{align*}
\| {\hat{\mu}}_{k} \| \leq \|  C_s\textstyle^{-1/2} \| \| {\hat{\mu}}_{k}' \| \leq \textstyle\frac{1}{\sqrt{\lambda_{\min}( C_{s} )}}( \sqrt{\lambda_{\max}( C_{s} ) N} \bar{\mu} + 2\sqrt{2\varepsilon} ).
\end{align*}

For the covariance bound, we first note that the Bures distance satisfies
$
B^{2}( \Sigma_{k}',{\hat{\Sigma}}_{k}') \geq \textstyle\big( \sqrt{\textstyle\text{tr}( \Sigma_{k}' )} - \sqrt{\textstyle\text{tr}( {\hat{\Sigma}}_{k}' )} \big)^{2},
$
hence
$
| \sqrt{\textstyle\text{tr}( \Sigma_{k}' )} - \sqrt{\text{tr}( {\hat{\Sigma}}_{k}' )} | \leq 2\sqrt{2\varepsilon}.
$

Next, we bound $\text{tr}(\Sigma_k')$. Since $\Sigma_k'=C^{1/2}\Sigma_kC^{1/2}$, we have
$\text{tr}( \Sigma_{k}' )= \text{tr}(C_s \Sigma_{k} )\leq \lambda_{\text{max}}(C_s) \text{tr}( {\Sigma} _k)$.
The true covariance matrix $\Sigma_k$ is block-diagonal with each block satisfying
$\text{tr}(\Sigma_k^{(i)})\leq\text{tr}(\bar\Sigma)$, so $\text{tr}(\Sigma_k)\leq N\text{tr}(\bar\Sigma)$. Thus, 
$\text{tr}( \Sigma_{k}' )\leq \lambda_{\text{max}}(C_s)N \text{tr}( \bar\Sigma)$. Consequently,
$
\textstyle\sqrt{\text{tr}( {\hat{\Sigma}}_{k}' )} \leq \textstyle\sqrt{\text{tr}( \Sigma_{k}' )} + \big| \textstyle\sqrt{\text{tr}( \Sigma_{k}' )} - \textstyle\sqrt{\text{tr}( {\hat{\Sigma}}_{k}' )} \big| \leq \textstyle\sqrt{N \text{tr}( C_{s} \bar{\Sigma} )} + 2\sqrt{2\varepsilon}.
$
Finally, transforming back to $\hat \Sigma_k=C^{-1/2}\hat\Sigma_k'C^{-1/2}$ , we obtain
$
\textstyle\text{tr}( {\hat{\Sigma}}_{k} ) \leq \frac{1}{\lambda_{\min}( C_{s} )}( \sqrt{\lambda_{\max}( C_{s} ) N \text{tr}( \bar{\Sigma} )} + 2\sqrt{2\varepsilon} )^{2}.
$
\end{proof}

\subsection{Proof of Proposition \ref{proposition:upper_bound_for_quadratic_cost}}\label{app:upper_bound_for_quadratic_cost}
\begin{proof}
Let $z_{i|k}^{*}$ and $z_{i|k}^{f}$ be the nominal states corresponding to the control sequence ${\bar{u}}_{k}^{*}$ with zero disturbance and the fixed‑gain law $\bar u_f$ respectively. At the initial time, $z_{0|k}^{*} - z_{0|k}^{f} = x_{k} - x_{k} = 0.$
For $i = 1,\cdots,N$, using linearity of the dynamics we obtain
\begin{align*}
z_{i|k}^{*} - z_{i|k}^{f} &= f( z_{i - 1|k}^{*},u_{i|k}^{*},0 ) - f( z_{i - 1|k}^{f},u_{f},0 )\\
&= 
 \textstyle\sum_{j = 1}^{i}{A^{i - j}B( u_{j - 1|k}^{*} - u_{f} )}.
\end{align*}
Taking norms and employing the definitions of $L_{A}$, $L_{B}$ yields
\begin{align*}
\| z_{i|k}^{*} - z_{i|k}^{f} \| \leq \textstyle\sum_{j = 1}^{i}{L_{A}^{i - j}L_{B}u_{u}}.  
\end{align*}
For $i = 1,\cdots,N - 1$, we have that
\begin{align*}
\| z_{i|k}^{*} \|_{Q}^{2} &\leq 2\| z_{i|k}^{f} \|_{Q}^{2} + 2\lambda_{\max}(Q)( \textstyle\sum_{j = 1}^{i}{L_{A}^{i - j}L_{B}u_{u}} )^{2},\\
\| z_{N|k}^{*} \|_{P}^{2} &\leq 2\| z_{N|k}^{f} \|_{P}^{2} + 2\lambda_{\max}(P)( \textstyle\sum_{j = 1}^{N}{L_{A}^{N - j}L_{B}u_{u}} )^{2},
\\
\| u_{i|k}^{*} \|_{R}^{2} &\leq 2\| u_{f} \|_{R}^{2} + 2\lambda_{\max}(R)u_{u}^{2}.
\end{align*}
Substituting these inequalities into the definition of $V_{q}\left( {\bar{u}}_{k}^{*},0 \right)$ yields
\begin{align*}
V_{q}( {\bar{u}}_{k}^{*},0 ) &= \textstyle\sum_{i = 0}^{N - 1}{l( z_{i|k}^{*},u_{i|k}^{*} )} + V_{f}( z_{N|k}^{*} )\\
&\leq 2V_{q}\left( {\bar{u}}_{f},0 \right) + c_{1}\leq 2\textstyle\frac{\lambda_{\max}(P)}{\lambda_{\min}(Q)}l( x_{k},u_{k} ) + c_{1}.
\end{align*}
Finally, writing the actual cost $V_{q}\left( {\bar{u}}_{k}^{*},{\bar{w}}_{k} \right)$ explicitly,
\begin{align*}
V_{q}( {\bar{u}}_{k}^{*},{\bar{w}}_{k} ) &= V_{q}( {\bar{u}}_{k}^{*},0 ) + \| \bar{A}x_{k} + \bar{B}{\bar{u}}_{k}^{*} + \bar{D}{\bar{w}}_{k} \|_{\bar{Q}}^{2} \\&- \left\| \bar{A}x_{k} + \bar{B}{\bar{u}}_{k}^{*} \right\|_{\bar{Q}}^{2}= V_{q}( {\bar{u}}_{k}^{*},0 ) + g_{1}( {\bar{w}}_{k}).
\end{align*}
Substituting the bound obtained in the previous step completes the proof of Proposition \ref{proposition:upper_bound_for_quadratic_cost}. 
\end{proof}

\subsection{Proof of Proposition \ref{proposition:recurrence_inequality__for_quadratic_cost}}\label{app:recurrence_inequality__for_quadratic_cost}
\begin{proof} 
We begin by examining the dynamics of the state deviation $\delta_{i} = x_{i|k + 1}^{c} - z_{i + 1|k}^{*}$. From the initial condition, we have $\delta_0=Dw_{k}$. For $i = 1,\cdots,N - 1$, recursive application of the linear dynamics yields the explicit expression,
\begin{align*}
\delta_{i}  &= f( x_{i - 1|k + 1}^{c},u_{i - 1|k + 1}^{c},w_{i - 1|k + 1} ) - f( z_{i|k}^{*},u_{i|k}^{*},0 )\\
& = A^{i}Dw_{k} + \textstyle\sum_{j = 0}^{i - 1}{A^{i - 1 - j}Dw_{j|k + 1}}.
\end{align*}
Next, we estimate the stage costs. For $\forall\epsilon > 0$ and $i = 1,\cdots,N - 2$, Young’s inequality gives
\begin{align*}
\textstyle\frac{1}{1 + \epsilon}\| x_{i|k + 1}^{c} \|_{Q}^{2} \leq \| z_{i + 1|k}^{*} \|_{Q}^{2} + \frac{1}{\epsilon}\| \delta_{i} \|_{Q}^{2}.
\end{align*}
Consequently, the stage cost satisfies,
\begin{align*}
&\textstyle\frac{1}{1 + \epsilon}l( x_{i|k + 1}^{c},u_{i|k + 1}^{c} ) = \frac{1}{1 + \epsilon}( \| x_{i|k + 1}^{c} \|_{Q}^{2} + \| u_{i|k + 1}^{c} \|_{R}^{2} )\\
&\leq \| z_{i + 1|k}^{*} \|_{Q}^{2} + \textstyle\frac{1}{\epsilon}\| \delta_{i} \|_{Q}^{2} + \frac{1}{1 + \epsilon}\| u_{i|k + 1}^{c} \|_{R}^{2}\\
&\leq l( z_{i + 1|k}^{*},u_{i + 1|k}^{*} ) + \textstyle\frac{1}{\epsilon}\| \delta_{i} \|_{Q}^{2}.
\end{align*}
Summing over $i = 1,\cdots,N - 2$ leads to
\begin{align*}
{\textstyle\frac{1}{1 + \epsilon}}\sum_{i = 0}^{N - 2}{l( x_{i|k + 1}^{c},u_{i|k + 1}^{c})} \leq \sum_{i = 1}^{N - 1}{l( z_{i|k}^{*},u_{i|k}^{*} )} + {\textstyle\frac{1}{\epsilon}}\sum_{i = 0}^{N - 2}\| \delta_{i} \|_{Q}^{2}.
\end{align*}
We now turn to the terminal step $i=N-1$. Note that the terminal control is $u_{N-1|k+1}^c=u_f$ and terminal state can be decomposed as
\begin{align*}
x_{N|k + 1}^{c} &= f( x_{N - 1|k + 1}^{c},u_{N - 1|k + 1}^{c},w_{N - 1|k + 1} )\\
&= f( x_{N - 1|k + 1}^{c},u_{f},0 ) + Dw_{N - 1|k + 1}.
\end{align*}
Utilizing the descent property \eqref{eq_uf_descent_property} of the terminal cost function,
\begin{align*}
V_{f}( f( x_{N - 1|k + 1}^{c},u_{f},0 ) ) &\leq V_{f}( x_{N - 1|k + 1}^{c} ) \\& - l( x_{N - 1|k + 1}^{c},u_{N - 1|k + 1}^{c} ),
\end{align*}
we can combine the terminal cost and the last stage cost as
\begin{align*}
V_{f}( x_{N|k + 1}^{c} ) + l( x_{N - 1|k + 1}^{c},u_{N - 1|k + 1}^{c} ) \leq V_{f}( x_{N - 1|k + 1}^{c} ) + \delta_{f}.
\end{align*}
where $\delta_{f}$ is an additional term arising from the disturbance.
Substituting all these estimates into the definition of $V_{q}( {\bar{u}}_{k + 1}^{c},{\bar{w}}_{k + 1} )$ and rearranging terms yields

\begin{align*}
&\textstyle\frac{1}{1 + \epsilon}V_{q}( {\bar{u}}_{k + 1}^{c},{\bar{w}}_{k + 1} ) 
\leq \sum_{i = 1}^{N - 1}{l( z_{i|k}^{*},u_{i|k}^{*} )} + \frac{1}{\epsilon}\sum_{i = 0}^{N - 2}\| \delta_{i} \|_{Q}^{2}\\
&+ \textstyle\frac{1}{1 + \epsilon}[V_{f}( x_{N - 1|k + 1}^{c} ) + \delta_{f}]\\
&\textstyle= V_{q}( {\bar{u}}_{k}^{*},0 ) - l( x_{k},u_{k} ) - V_{f}( z_{N|k}^{*} ) + \frac{1}{\epsilon}\sum_{i = 0}^{N - 2}\| \delta_{i} \|_{Q}^{2}\\
& \textstyle+ \frac{1}{1 + \epsilon}[V_{f}( x_{N - 1|k + 1}^{c} ) + \delta_{f}].
\end{align*}
Finally, observing that $x_{N - 1|k + 1}^{c} - z_{N|k}^{*} = \delta_{N - 1}$ and expanding its terminal cost,
\begin{align*}
V_{f}( x_{N - 1|k + 1}^{c} ) = V_{f}( z_{N|k}^{*} + \delta_{N - 1} )= V_{f}( z_{N|k}^{*} ) + g_{2}\left( \delta_{N - 1} \right),
\end{align*}
we obtain the inequality stated in the proposition. 
\end{proof}

\subsection{Proof of Proposition \ref{propositopn:conditional_expectation_term_bound}}\label{app:conditional_expectation_term_bound}
\begin{proof}
We analyze each term on the left-hand side of inequality \eqref{eq_disturbed_terms} individually. The analysis relies on bounding sums of the form $\sum L_A^j$ arising from the expansion of system dynamics. Under Assumption \ref{assumption_dynamic_bound}, the following finite constants are defined to simplify subsequent expressions
\begin{align*}
C_{A1} &= \sqrt{\textstyle\sum_{j=0}^{N-2} L_A^j} ,
C_{A2} = \sqrt{\textstyle\sum_{i=0}^{N-2} L_A^{2i}} , \\
C_{A3} &= \sqrt{\textstyle\sum_{i=0}^{N-2} (\sum_{j=0}^{i-1} L_A^{j})^2},
C_{A4} = \sqrt{\textstyle\sum_{i=1}^N L_A^{2i}}, \\
C_{A5} &= \sqrt{\textstyle\sum_{i=1}^N \sum_{j=0}^{i-1} L_A^{2j}}.
\end{align*}
These constants $C_{A1}-C_{A5}$ are used throughout the proof to aggregate the effects of the disturbance propagation over the prediction horizon. Then, we have

1) the terminal state deviation $\delta_f$ defined in \eqref{eq_delta_f},
\begin{align*}
&\mathbb{E}_{\mathbb{P}_{k+1}^\ast}  [\delta_f| \mathcal{F}_k] = \mathbb{E}_{\mathbb{P}_{k+1}^\ast} [\|D w_{N-1|k+1}\|_P^2 | \mathcal{F}_k] \\
    &+ 2 (z_{N|k}^{\ast}+A^{N-1} D w_k)^T P D \mathbb{E}_{\mathbb{P}_{k+1}^\ast} [w_{N-1|k+1} | \mathcal{F}_k] \\
    &+ 2(\textstyle\sum_{j=0}^{N-2}A^{N-2-j} D \mathbb{E}_{\mathbb{P}_{k+1}^\ast} [w_{j|k+1} | \mathcal{F}_k])^T P D \\&\quad\mathbb{E}_{\mathbb{P}_{k+1}^\ast} [w_{N-1|k+1} | \mathcal{F}_k] \\
    &\le \lambda_{\max}(P) L_D^2 (\|\hat\mu_{k+1}\|^2 + \text{tr}(\hat\Sigma_{k+1})) \\
    &\quad + 2 \lambda_{\max}(P) L_D \|z_{N|k}^\ast\| \|\hat\mu_{k+1}\| + \|w_k\| L_A^{N-1} L_D \|\hat\mu_{k+1}\| \\
    &\quad + 2 \lambda_{\max}(P) L_D^2 \textstyle\sum_{j=0}^{N-2} L_A^j (\|\hat\mu_{k+1}\|^2 + \text{tr}(\hat\Sigma_{k+1})) \\
    &\le \lambda_{\max}(P) L_D^2 (1 + 2C_{A1}) (\|\hat\mu_{k+1}\|^2 + \text{tr}(\hat\Sigma_{k+1})) \\
    &\quad + 2 \lambda_{\max}(P) L_D \|z_{N|k}^\ast\| \|\hat\mu_{k+1}\| + L_A^{N-1} L_D \|w_k\| \|\hat\mu_{k+1}\|.
    \end{align*}

2) the state deviation sum of $\delta_i$ defined in \eqref{eq_delta_i}, 
\begin{align*}
    &\mathbb{E}_{\mathbb{P}_{k+1}^\ast}  [\textstyle \sum_{i=0}^{N-2} \|\delta_i\|_Q^2 | \mathcal{F}_k ] \\
    &=\textstyle \sum_{i=0}^{N-2} \mathbb{E}_{\mathbb{P}_{k+1}^\ast} [ \| A^i D w_k + \sum_{j=0}^{i-1} A^{i-1-j} D w_{j|k+1} \|_Q^2 | \mathcal{F}_k ] \\
    &\le \lambda_{\max}(Q) L_D^2 \big( 2 C_{A2}^2 \|w_k\|^2 + C_{A3}^2 (\|\hat\mu_{k+1}\|^2 + \text{tr}(\hat\Sigma_{k+1})) \big).
    \end{align*}

3) the term $g_2$ defined in \eqref{eq_g2},
    \begin{align*}
    &\mathbb{E}_{\mathbb{P}_{k+1}^\ast}  [ g_2(\delta_{N-1}) | \mathcal{F}_k ] \\
    &\textstyle= \mathbb{E}_{\mathbb{P}_{k+1}^\ast} [ \| A^{N-1} D w_k + \sum_{j=0}^{N-2} A^{N-2-j} D w_{j|k+1} \|_P^2 | \mathcal{F}_k ] \\
    &\textstyle\quad + 2 z_{N|k}^{\ast T} P \mathbb{E}_{\mathbb{P}_{k+1}^\ast} [ A^{N-1} D w_k + \sum_{j=0}^{N-2} A^j D w_{j|k+1} | \mathcal{F}_k ] \\
    &\le \lambda_{\max}(P) L_D^2 \big( L_A^{N-1} \|w_k\| + C_{A1}^2 \|\hat\mu_{k+1}\|^2 + C_{A2}^2 \text{tr}(\hat\Sigma_{k+1}) \big) \\
    &\quad + 2 \|z_{N|k}^\ast\| \lambda_{\max}(P) L_D ( L_A^{N-1} \|w_k\| + C_{A1}^2 \|\hat\mu_{k+1}\| ).
    \end{align*}

4) the term $g_1$ defined in \eqref{eq_g1},
\begin{align*}
    \mathbb{E}_{\mathbb{P}_k^\ast} & [g_1(\bar{w}_k) | \mathcal{F}_{k-1}] = \mathbb{E}_{\mathbb{P}_k^\ast} [\|\bar D\bar w_k\|^2_{\bar Q} | \mathcal{F}_{k-1}] \\
    &\quad + 2 (\bar A x_k + \bar B\bar u_k^\ast)^T\bar Q \bar D \mathbb{E}_{\mathbb{P}_k^\ast} [\bar w_k | \mathcal{F}_{k-1}] \\
    &\le \lambda_{\max}(Q) L_D^2 \textstyle\sum_{i=1}^{N-2} \sum_{j=0}^{i-1} L_A^{2(i-j)} \mathbb{E}_{\mathbb{P}_k^\ast} [\|w_{j|k}\|^2 | \mathcal{F}_{k-1}] \\
    &\quad + \lambda_{\max}(P) L_D^2 L_A^{2(N-1)} \mathbb{E}_{\mathbb{P}_k^\ast} [\|w_{N-1|k}\|^2 | \mathcal{F}_{k-1}] \\
    &\quad + 2 (\bar A x_k + \bar B\bar u_k^\ast)^T\bar Q\bar D \mathbb{E}_{\mathbb{P}_k^\ast} [\bar w_k | \mathcal{F}_{k-1}] \\
    &\le \lambda_{\max}(P) L_D^2 C_{A5}^2 (\|\mu_k\|^2 + \text{tr}(\hat\Sigma_k)) \\
    &\quad + 2 \lambda_{\max}(P) L_D C_{A4} \left( C_{A4} \|x_k\| + 0.5u_u L_B C_{A5} \right) \|\hat\mu_k\|.
\end{align*}

5) bounding of constraint-related terms,

Let $C_{w1} = \lambda_{\max}(F_0^T F_0) L_D^2 C_{A5}^2$, then
\begin{align*}
    \mathbb{E}_{\mathbb{P}_{k+1}^\ast} & [\|F \bar{D} \bar{w}_{k+1}\|^2 | \mathcal{F}_k] 
    \le \lambda_{\max}(F_0^T F_0) L_D^2 \\&\times\textstyle\sum_{i=0}^{N-1} \sum_{j=0}^i L_A^{2(i-j)} \mathbb{E}_{\mathbb{P}_{k+1}^\ast} [\|w_{j|k+1}\|^2 | \mathcal{F}_k] \\
    &\le C_{w1} (\|\hat\mu_{k+1}\|^2 + \text{tr}(\hat\Sigma_{k+1})).
\end{align*}
Similarly,
 \begin{equation*}
 \mathbb{E}_{\mathbb{P}_k^\ast} [\|F \bar{D} \bar{w}_k\|^2 | \mathcal{F}_{k-1}] \le C_{w1} (\|\hat\mu_k\|^2 + \text{tr}(\hat\Sigma_k)).
\end{equation*}

6) cross-term bounding using Young's inequality,
    \begin{align*}
    & 2\lambda_{\max}(P)L_D(C_{A1}^2 + 1)\|z_{N|k}^\ast\|\|\hat\mu_{k+1}\| \le (\textstyle\frac{1}{12l_c\lambda_{\min}(Q)}\|z_{N|k}^\ast\|^2 \\
    & \quad + \textstyle\frac{12l_c\lambda_{\max}^2(P)}{\lambda_{\min}(Q)}L_D^2(C_{A1}^2 + 1)^2\|\hat\mu_{k+1}\|^2), \\
    & 2\lambda_{\max}(P)L_DL_A^{N-1}\|z_{N|k}^\ast\|\|w_k\| \le (\textstyle\frac{1}{12l_c\lambda_{\min}(Q)}\|z_{N|k}^\ast\|^2 \\
    & \quad + \textstyle\frac{12l_c\lambda_{\max}^2(P)}{\lambda_{\min}(Q)}L_D^2L_A^{2(N-1)}\|w_k\|^2), \\
    & \lambda_{\max}(P)L_D^2L_A^{N-1}(1 + C_{A1}^2)\|w_k\|\|\hat\mu_{k+1}\| \le (\lambda_{\max}(P)L_D^2L_A^{N-1} \\
    & \quad \times (1 + C_{A1}^2)(\|\hat\mu_{k+1}\|^2 + \|w_k\|^2)), \\
    & 2\lambda_{\max}(P)L_DC_{A4}^2\|x_k\|\|\hat\mu_k\| \le (\textstyle\frac{1}{12}\lambda_{\min}(Q)\|x_k\|^2 \\
    & \quad + 12\textstyle\frac{\lambda_{\max}^2(P)}{\lambda_{\min}(Q)}L_D^2C_{A4}^4\|\hat\mu_k\|^2).
    \end{align*}

7) incorporating the terminal constraint,

Based on the terminal constraints \eqref{eq_terminal_constraints}, we incorporate the term $\|z_{N|k}^\ast\|^2$ into the $\|x_k\|^2$ term, 
\begin{equation*}
\textstyle\frac{1}{6l_c}\lambda_{\min}(Q)\|z_{N|k}^\ast\|^2 \le \frac{1}{6}\lambda_{\min}(Q)\|x_k\|^2.
\end{equation*}
From this, we obtain the coefficient for the state term,
$
k_0 = \textstyle\frac{1}{4}\lambda_{\min}(Q).
$
Finally, combining all the bounds derived in steps 1-7, the overall inequality is established as \eqref{eq_disturbed_terms},
where the constants are defined as follows,
\begin{align*}
    k_0 &= \textstyle\frac{1}{4}\lambda_{\min}(Q), \\
    k_1 &= 12l_c\textstyle\frac{\lambda_{\max}^2(P)}{\lambda_{\min}(Q)}L_D^2 L_A^{2(N-1)}  + \lambda_{\max}(P)L_D^2 L_A^{2(N-1)} \\
    &\quad +\lambda_{\max}(P)L_D^2 L_A^{N-1}(1+C_{A1}^2)+ 2\lambda_{\max}(Q)L_D^2 C_{A2}^2 , \\
    k_2 &= 12l_c\textstyle\frac{\lambda_{\max}^2(P)}{\lambda_{\min}(Q)}L_D^2 (C_{A1}^2 + 1)^2 + 2\lambda_{\max}(Q)L_D^2 C_{A3}^2 \\
    &\quad + \lambda_{\max}(P)L_D^2 (1 + 2C_{A1} + C_{A1}^4)  \\
    &\quad + \lambda_{\max}(P)L_D^2 L_A^{N-1}(1+C_{A1}^2) + \textstyle\frac{1}{4\epsilon_{c2}}C_{w1}, \\
    k_{31} &= \lambda_{\max}(P)\lambda_{\min}(Q)L_D C_{A4} u_u L_B C_{A5}, \\
    k_{32} &= \lambda_{\max}(P)L_D^2 C_{A5}^2 + \textstyle\frac{1}{4\epsilon_{c1}}C_{w1} + 12\textstyle\frac{\lambda_{\max}^2(P)}{\lambda_{\min}(Q)}L_D^2 C_{A4}^4, \\
    k_4 &= \lambda_{\max}(Q)L_D^2 C_{A3}^2 + \lambda_{\max}(P)L_D^2 C_{A2}^2 + \textstyle\frac{1}{4\epsilon_{c2}}C_{w1}  \\
    &\quad +\lambda_{\max}(P)L_D^2 (1+2C_{A1}), \\
    k_5 &= \lambda_{\max}(P)L_D^2 C_{A5}^2 + \textstyle\frac{1}{4\epsilon_{c1}}C_{w1}.
\end{align*}

\end{proof}

\section*{References}
\bibliographystyle{IEEEtran}
\bibliography{IEEEabrv,mybibfile.bib}

\end{document}